\documentclass[letter]{article}

\usepackage{amsthm}
\usepackage{times}
\usepackage{url}
\usepackage{epsfig}
\usepackage{amsmath}
\usepackage{longtable}
\usepackage{amsfonts}
\usepackage{amssymb}
\usepackage{graphicx}
\usepackage{epstopdf}
\usepackage{multirow}
\usepackage{subcaption}
\usepackage{verbatim}
\usepackage{lineno}
\usepackage{enumitem}
\usepackage{color}
\usepackage{hyperref}
\usepackage{xspace}

\usepackage[linesnumbered,ruled]{algorithm2e}

\usepackage{self}

\newtheorem{proposition}{Proposition}

\newtheorem{example}{Example}

\usepackage{hyperref}
\usepackage{xcolor}
\hypersetup{
	colorlinks,
	linkcolor={red!50!black},
	citecolor={blue!50!black},
	urlcolor={blue!80!black}
}

\usepackage[numbers]{natbib}

\graphicspath{{img2/}}

\begin{document}

\title{Estimating multi-class dynamic origin-destination demand through a forward-backward algorithm on  computational graphs}

\author{Wei Ma, Xidong Pi, Zhen (Sean) Qian\\
	Department of Civil and Environmental Engineering\\ Carnegie Mellon University, Pittsburgh, PA 15213\\
	\{weima, xpi, seanqian\}@cmu.edu}
\maketitle

\begin{abstract}

Transportation networks are unprecedentedly complex with heterogeneous vehicular flow. Conventionally, vehicle classes are considered by vehicle classifications (such as standard passenger cars and trucks). However, vehicle flow heterogeneity stems from many other aspects in general, e.g., ride-sourcing vehicles versus personal vehicles, human driven vehicles versus connected and automated vehicles. Provided with some observations of vehicular flow for each class in a large-scale transportation network,  how to estimate the multi-class spatio-temporal vehicular flow, in terms of time-varying Origin-Destination (OD) demand and path/link flow, remains a big challenge.  This paper presents a solution framework for multi-class dynamic OD demand estimation (MCDODE) in large-scale networks. The proposed framework is built on a computational graph with tensor representations of spatio-temporal flow and all intermediate features involved in the MCDODE formulation. A forward-backward algorithm is proposed to efficiently solve the MCDODE formulation on computational graphs. In addition, we propose a novel concept of tree-based cumulative curves to estimate the gradient of OD demand. A Growing Tree algorithm is developed to construct tree-based cumulative curves. The proposed framework is examined on a small network as well as a real-world large-scale network. The experiment results indicate that the proposed framework is compelling, satisfactory and computationally plausible.  

\end{abstract}

\section{Introduction}

Transportation networks are unprecedentedly complex with heterogeneous vehicular flow. Conventionally, vehicle heterogeneity are considered by vehicle classifications (such as standard passenger cars and trucks). However, vehicle flow heterogeneity stems from many other aspects in general, e.g., ride-sourcing vehicles versus personal vehicles, human driven vehicles versus connected and automated vehicles. How to effectively manage the multi-class vehicles in a complex transportation system so as to improve the network efficiency presents a big challenge. As an indispensable component of dynamic transportation network models  with heterogeneous traffic, the multi-class dynamic origin-destination (OD) demand plays a key role in transportation planning and management to understand spatio-temporal vehicular flow and its travel behavior. To our best knowledge, there is a lack of studies to understand and estimate the dynamic OD demand of multiple vehicle classes using sparse and partial flow observations. In view of this, this paper presents a data-driven framework for multi-class dynamic OD demand estimation (MCDODE) in large-scale networks. The MCDODE formulation is represented on a computational graph, and a novel forward-backward algorithm is proposed to estimate the OD demand efficiently and effectively. The proposed MCDODE framework is examined in a small networks as well as a real-world large-scale network to demonstrate the estimation accuracy and the computational efficiency.

The multi-class dynamic OD demand (MCDOD) represents the number of vehicles in each of the vehicle classes ({\em e.g.} personal cars, trucks, ride-sourcing vehicles) departing from an origin and heading to a destination in a particular time interval. The definition of ``classes'' is very general, by vehicle sizes, specifications, and nature of trips. The MCDOD reveals the fine-gained traffic demand information for different vehicle classes and the overall spatio-temporal mobility patterns can be inferred from the OD demand and its resultant path/link flows.  Policymakers can understand the departure/arrival patterns of multi-class vehicles through the MCDOD. The MCDOD also helps the policymakers understand the impact of each vehicle class on the roads, and hence each class of vehicles can be managed separately. In addition, most of Advanced Traveler Information Systems/Advanced Traffic Management Systems (ATIS/ATMS) would require accurate MCDOD as the model input \citep{huang2007multiclass}.

The dynamic OD estimation (DODE) has been extensively studied over the past few decades. A generalized least square (GLS) formulation is proposed for estimating dynamic OD demand with exogenous route choice model \citep{cascetta1993dynamic}. The GLS formulation is further extended to a bi-level optimization problem, and the bi-level formulation solves for the DODE with endogenous route choice model \citep{tavana2001internally}. Advantages and disadvantages of the bi-level formulation  are discussed in \cite{nguyen1977estimating, leblanc1982selection, fisk1989trip, yang1992estimation, florian1995coordinate, jha2004development}. The Stochastic Perturbation Simultaneous Approximation (SPSA) methods have also been adopted to solve the bi-level formulation in many studies \citep{balakrishna2008time, cipriani2011gradient, lee2009new, vaze2009calibration, ben2012dynamic, lu2015enhanced, tympakianaki2015c, antoniou2015w}.  \citet{zhang2017efficient, osorio2019dynamic} proposed an instantaneous approximation for the dynamic traffic assignment models, and the DODE problem is cast into a convex optimization which can be solved efficiently. The bi-level formulation can also be relaxed to a single-level problem and solved by updating the OD demand using gradient-based methods \citep{nie2008variational, lu2013dynamic}.


The DODE problem can be viewed as a statistical estimation problem. For example, a statistical inference framework using Markov Chain Monte Carlo algorithm was proposed to estimate the probabilistic OD demand \citep{hazelton2008statistical}. \citet{zhou2003dynamic} proposed a DODE framework with multi-day data and established a hypothesis testing framework to identify the demand evolution within a week.
In addition to the off-line DODE methods, the DODE formulation can also be solved with real-time data streams for ATIS/ATMS applications. \citet{bierlaire2004efficient} extended the GLS formulation on large-scale networks and solved the formulation efficiently on a real-time basis. The state-space models are also used to estimate the dynamic OD  of each time interval sequentially \citep{zhou2007structural, ashok2000alternative}.

The DODE methods often encapsulate the dynamic traffic assignment (DTA) models within the bi-level formulation. Various DTA models can be adopted to model the network dynamics and travelers' behaviors. For example, Dynamic Network Loading (DNL) models \citep{ma2008polymorphic} simulate the vehicle trajectories and traffic conditions given determined the origin-destination (O-D) demand and fixed routes; Dynamic User Equilibrium (DUE) models \citep{mahmassani1984dynamic, nie2010solving} search for the equilibrated travelers' route choices to achieve user optima; Dynamic System Optimal (DSO) models \citep{shen2007path, qian2012system, ma2014continuous} explore the optimal conditions under which the total network costs are minimized.

The multi-class traffic assignment models and OD estimation models can be built by extending the single-class models. However, the main challenges are how to estimate MCDOD in such a way to match multi-source spatio-temporal data in a large-scale transportation network. A number of studies \citep{gundaliya2008heterogeneous, venkatesan2008development, qian2017modeling} investigate the multi-class traffic flows in general networks. The multi-class DTA models are studied by \citet{dafermos1972traffic, yang2004multi, huang2007multiclass} without the real-world data validation. There are studies estimating the static OD demand for multiple vehicle classes \citep{wong2005estimation, raothanachonkun2006estimation, noriega2007multi, zhao2018multiclass}. In particular, \citet{shao2015estimation} estimated the probabilistic distribution of the multi-class OD using traffic counts. The research gap lies in estimating MCDOD in large-scale networks that are consistent with real-world multi-source traffic data.

In this paper, we develop a data-driven framework that estimates the multi-class dynamic OD demand using traffic counts and travel time data that are partially observed in general networks. The proposed framework formulates the MCDODE problem and represents it with a computational graph. The MCDODE is solved with a novel forward-backward algorithm on the computational graph. The MCDODE framework can be solved using multi-core CPUs or Graphics Processing Units (GPUs), and hence the proposed method can be computational efficient and applied to the large-scale networks and multi-day data.  The closest work to this paper is that of \citet{wu2018hierarchical}, which constructs a layered computational graph for the single-class static OD estimation problem. This paper extends the computational graph approach to the case of multi-class dynamic OD demand by tackling additional challenges on the incorporating flow dynamics and characteristics across both classes and time of day. We also build a novel framework to evaluate the gradients of multi-class OD demand for simulation-based traffic assignment models. The main contributions of this paper are summarized as follows:
\begin{enumerate}[label=\arabic*)]
	\item We propose a  theoretical formulation for estimating multi-class dynamic OD demand. The formulation is represented on a computational graph such that the MCDODE can be solved for large-scale networks with large-scale traffic data. The proposed MCDODE formulation can handle any form of traffic data, such as flow, speed or trip cost.
	\item We propose a novel forward-backward algorithm to solve for the MCDODE formulation on the constructed computational graph with simulation-based traffic assignment models.
	\item We use a tree-based cumulative curves to evaluate the gradient of multi-class dynamic OD demand in the forward-backward algorithm, and a Growing Tree algorithm is proposed to construct the tree-based cumulative curves during the traffic simulation.
	\item We examine the proposed MCDODE framework on a large-scale network to demonstrate its effectiveness and computational efficiency of the solution algorithm.
\end{enumerate}

The remainder of this paper is organized as follows.  Section \ref{sec:formulation} presents the formulation details of MCDODE. Section \ref{sec:solution} describes the solution pipeline and discusses practical issues for MCDODE framework. In section~\ref{sec:exp}, a small network is used to demonstrate the accuracy, efficiency and robustness of the MCDODE framework. In addition, we demonstrate the scalability and computational efficiency of the proposed framework using real-world data in a large-scale network. At last, conclusions are drawn in section \ref{sec:con}.

\section{Formulation}
\label{sec:formulation}
In this section, we discuss the multi-class dynamic origin-destination estimation (MCDODE) framework. We first present the notations used in this paper, and then the multi-class dynamic network flow  is modeled. Secondly, the MCDODE problem is formulated and represented on a computational graph. We then solve the MCDODE through a novel forward-backward algorithm.

\subsection{Notations}
Notations are summarized in Table~\ref{tab:notation}.
\begin{longtable}{p{3cm}p{10cm}}
	\caption{\footnotesize List of notations}
	\label{tab:notation}
	\endfirsthead
	\endhead	
	$A$ & The set of all links\\
	$K_q$ & The set of all OD pairs\\
	$K_{rs}$ & The set of all paths between OD pair $rs$\\
	$B$ & The set of indices of the observed flow\\
	$E$ & The set of indices of the observed travel time\\
	$H$ & The set of all time intervals\\
	$D$ & The set of vehicle classes\\
	\multicolumn{2}{c}{\textbf{Variables as scalars}}\\
	$h_1$ & The index of departure time interval of path flow or OD flow\\
	$h_2$ & The index of arrival time interval at the tail of link\\
	$i$ & index of vehicle class\\
	
	$x_{ai}^{h_2}$ & The flow arriving at the tail of link $a$ in time interval $h_2$ for vehicle class $i$\\
	$t_{ai}^{h_2}$ & The link travel time of link $a$ in time interval $h_2$ for vehicle class $i$\\
	$q_{rsi}^{h_1}$ & The flow of OD pair $rs$ in time interval $h_1$ for vehicle class $i$\\
	$f_{rsi}^{kh_1}$ & The $k$th path flow for OD pair $rs$ in time interval $h_1$ for vehicle class $i$\\
	$p_{rsi}^{kh_1}$ & The route choice portion of choosing path $k$ in all paths between OD pair $rs$ in time interval $h_1$ for vehicle class $i$\\
	$c_{rsi}^{kh_1}$ & The path travel time of $k$th path for OD pair $rs$ in time interval $h_1$ for vehicle class $i$\\	
	$L_{ai}^{bh_2}$ & The observation/link incidence for link $a$ in time interval $h_2$ and observed flow $b$ for vehicle class $i$\\
	$M_{ai}^{eh_2}$ & The weight of travel time for $t_{ai}^{h_2}$ in the observed travel time $e$\\
	$\rho_{rsi}^{ka}(h_1, h_2)$ & The portion of the $k$th path flow departing within time interval $h_1$ between OD pair $rs$ which arrives at link $a$ within time interval $h_2$  for vehicle class $i$ (namely, an entry of the DAR matrix)\\
	$y_b$ & The $b$th observed flow, which might be a linear combination of link flow across vehicle classes, road segments and over time intervals\\
	$z_e$ & The $e$th observed travel time, which might be a linear combination of link travel time across vehicle classes, road segments and over time intervals\\
	\multicolumn{2}{c}{\textbf{Variables as tensors}}\\
	$\vec{y}$ & The vector of the observed flow\\
	$\vec{z}$ & The vector of the observed travel time\\
	$\vec{x}_i$ & The vector of link flow for vehicle class $i$\\
	$\vec{t}_i$ & The vector of link travel time for vehicle class $i$\\
	$\vec{f}_i$ & The vector of path flow for vehicle class $i$\\
	$\vec{c}_i$ & The vector of path travel time for vehicle class $i$\\
	$\vec{L}_i$ & The observation/link incidences matrix for vehicle class $i$\\
	$\vec{M}_i$ & The travel time weight matrix for vehicle class $i$\\
	$\vec{p}_i$ & The matrix of route choice portions for vehicle class $i$\\
	$\mathbf{\rho}_i$ & The dynamic assignment ratio (DAR) matrix for vehicle class $i$\\
\end{longtable}

\subsection{Modeling multi-class dynamic network flow}
In this section, we first formulate a general network model for multi-class traffic flow in discrete time. We denote the path flow $f_{rsi}^{kh_1}$ as the number of class-$i$ vehicles  departing from $k$th path for OD pair $rs$ in time interval $h_1$ and link flow $x_{ai}^{h_2}$ as the number of class-$i$ vehicles arriving at the tail of link $a$ in time interval $h_2$. The relationship between path flow and link flow is presented in Equation~\ref{eq:linkpath}.

\begin{equation}
\label{eq:linkpath}
x_{ai}^{h_2}
=\sum_{rs \in K_q} \sum_{k \in K_{rs}} \sum_{h_1 \in H}   \rho_{rsi}^{ka}(h_1, h_2) f_{rsi}^{kh_1}
\end{equation}
where $K_q$ is the set of all OD pairs, and $K_{rs}$ is the path set for OD pair $rs$. $H$ is the set of all possible time intervals during study period. To avoid confusion, we denote departure time interval of path flow or OD flow as $h_1$, and the arrival time interval at the tail of link as $h_2$, respectively. The dynamic assignment ratio (DAR) $\rho_{rsi}^{ka}(h_1, h_2)$ denotes the portion of the $k$th path flow departing within time interval $h_1$ for OD pair $rs$ which arrives at link $a$ within time interval $h_2$ for vehicle class $i$ \citep{ma2018estimating}. For each OD pair $rs$, there are $K_{rs}$ paths for travelers to choose from, and the portion of choosing path $k$ in time interval $h_1$ for vehicle class $i$ is denoted as $p_{rsi}^{kh_1}$. The OD flow and path flow for vehicle class $i$ can be represented by Equation~\ref{eq:flow}.
\begin{equation}
\label{eq:flow}
f_{rsi}^{kh_1} =  p_{rsi}^{kh_1} q_{rsi}^{h_1}
\end{equation}
where OD demand $q_{rsi}^{h_1}$ represents the number of class-$i$ vehicles for OD pair $rs$ in time interval $h_1$. The link travel time, path travel time and DAR can be obtained from the dynamic network loading (DNL) models, as presented in Equation~\ref{eq:dnl}.

\begin{eqnarray}
\label{eq:dnl}
\left\{t_{ai}^{h_2}, c_{rsi}^{kh_1}, \mathbf{\rho}_{rsi}^{ka}(h_1, h_2) \right\} _{r,s,i,k,a,h_1, h_2} =  \Lambda(\{f_{rsi}^{kh_1}\}_{i,r,s,k,h_1})
\end{eqnarray}
where $\Lambda(\cdot)$ represents the multi-class dynamic network loading models. The travel time can be generalized to any form of the disutility as long as it can be simulated by $\Lambda$. The generalized travel time can include roads tolls, left turn penalty, travelers' preferences and so on.
The DNL function $\Lambda$ takes the multi-class path flow as input and outputs the spatio-temporal network conditions $\{t_{ai}^{h_2}, c_{rsi}^{kh_1} \}$ and the DAR matrix $\mathbf{\rho}_{rsi}^{ka}(h_1, h_2)$. Though there exists analytical solutions in small networks, $\Lambda$ is usually represented by simulation-based models in large-scale networks. Many existing models including, but are not limited to, DynaMIT \citep{ben1998dynamit}, DYNASMART \citep{mahmassani1992dynamic}, DTALite \citep{zhou2014dtalite} and MAC-POSTS \citep{qian2016dynamic, CARTRUCK}, can be potentially used as function $\Lambda$.

The route choice portion $p_{rsi}^{kh_1}$ is obtained from a generalized route choice function, as presented in Equation~\ref{eq:gen_choice}.

\begin{equation}
\label{eq:gen_choice}
p_{rsi}^{kh_1}= \Psi_{rsi}^k\left( \tilde{\vec{c}}_{rsi}^{k h_1}, \tilde{\vec{t}}_{ai}^{h_1} \right)
\end{equation}

\begin{eqnarray}
\tilde{\vec{c}}_{rsi}^{k h_1} = \left\{ c_{rsi}^{k h'}| h' \leq h_1, h' \in H, rs\in K_q, k \in K_{rs} , i \in D\right\}\\
\tilde{\vec{t}}_{ai}^{h_1} = \left\{ t_{ai}^{h'}| h' \leq h_1, h' \in H, a \in A , i \in D\right\}
\end{eqnarray}
where $\Psi_{rsi}^{kh_1}(\cdot. \cdot)$ is a generalized route choice model that takes in the path travel time and link travel time before time interval $h_1$ and computes the route choice portion for travelers in $k$th path of OD $rs$ for vehicle class $i$. We further denote $t_{ai}^{h_2}$ as the travel time of link $a$ in time interval $h_2$ for vehicle class $i$, and $c_{rsi}^{kh_1}$ as the travel time of path flow $k$ in OD pair $rs$ departing in time interval $h_1$. Most of the state-of-art route choice models \citep{prashker2004route, zhou2012c}, including Logit and Probit models, satisfy Equation~\ref{eq:gen_choice}. The Logit and Probit models have been adopted in many studies and achieved great success in real-world applications \citep{maher2001bi}. Combining Equation~\ref{eq:linkpath} and \ref{eq:flow}, we present the relation of link flow and OD flow in Equation~\ref{eq:linkod}.
\begin{equation}
\label{eq:linkod}
x_{ai}^{h_2} =\sum_{rs \in K_q} \sum_{k \in K_{rs}} \sum_{h_1 \in H}   \rho_{rsi}^{ka}(h_1, h_2) p_{rsi}^{kh_1} q_{rsi}^{h_1}
\end{equation}

\subsection{Modeling the observed flow and observed travel time}
In this section, we describe the concept of the observed flow and observed travel time. The motivation for the definition of the observed flow and travel time lies in the indirect and aggregated observations of traffic networks. For example, loop detectors measure the traffic counts in each time interval, but the vehicle classes cannot be differentiated. The Department of Transportation regularly hire individuals to count the vehicles on road segments for both directions and different vehicle classes. In this case, the observation of traffic flow does not differentiate road directions.

To accommodate different kinds of flow observations, we propose the concept of observed flow, denoted by $y_b$, as a linear combination of link flow across vehicle classes, road segments and time intervals. The formulation of the observed flow is presented in Equation~\ref{eq:fo}.

\begin{equation}
\label{eq:fo}
y_{b} = \sum_{i \in D} \sum_{a \in A} \sum_{h_2 \in H}  L_{ai}^{bh_2}  x_{ai}^{h_2}
\end{equation}
where $b \in B$ is the index of observed flow and $B$ is the set of indices of observed flow. The unit of $y_{b}$ is the same as $x_{ai}^{h_2}$. For example, if there are two loop detectors recording the traffic counts in each interval, then we have $|B| = 2|H|$. By formulation~\ref{eq:fo}, the observed flow can be the traffic count observations in various forms, including traffic counts of a single-class vehicles, aggregated traffic counts of all vehicles, aggregated traffic counts of a road for both directions, and aggregated traffic counts of multiple links. The link/observation incidence $L_{ai}^{bh_2}$ represents how the observed flow is aggregated. $L_{ai}^{bh_2}$ is $1$ if link flow $x_{ai}^{h_2}$ is observed in the observed flow $y_b$ and $0$ otherwise, and a similar definition can be found in \citet{yang2018stochastic}.

Similarly, we may also observe the travel time across multiple links, vehicle classes and time intervals. For example, we may observe the total travel time of a highway which consists of multiple consecutive links, or we may observe the average travel time of car and trucks in a single link. To accommodate different kinds of travel time observations, we assume the observed travel time can be represented by a linear combination of all link travel time, as presented in Equation~\ref{eq:tl}.

\begin{eqnarray}
\label{eq:tl}
z_e = \sum_{i \in D} \sum_{a \in A} \sum_{h_2 \in H} M_{ai}^{eh_2} t_{ai}^{h_2}
\end{eqnarray}
where $M_{ai}^{eh_2}$ represents the weight of travel time for $t_{ai}^{h_2}$ in the observed travel time $e$. The unit of $z_e$ is the same as $t_{ai}^{h_2}$. $e\in E$ is the index of observed travel time and $E$ is the set of indices of observed travel time. The formulation~\ref{eq:tl} is also general enough to accommodate various types of travel time observations. For example, the observed travel time includes link travel time of a single-class vehicles, path travel time of single-class vehicles, average travel time of all vehicle classes. We note that $M_{ai}^{eh_2} \in [0,1]$  since we may observe the average travel time of multiple links, while $L_{ai}^{bh_2} \in \{0, 1\}$ because traffic flow is usually observed in an aggregated manner.

\begin{example}
	\label{ex:1}
	To illustrate the formulation of the observed flow and observed travel time, we consider a two-link network presented in Figure~\ref{fig:n2}. We only consider one time interval, hence $|H| = 1$. We assume vehicle class $1$ represents cars, and vehicle class $2$ represents trucks.
	\begin{figure}[h]
		\centering
		\includegraphics[width=0.55\linewidth]{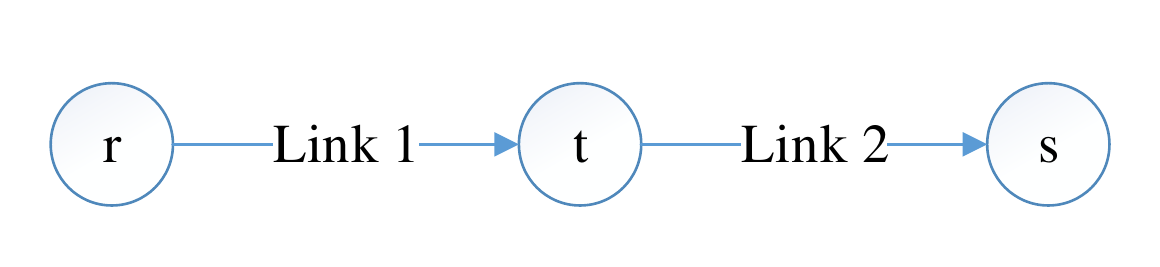}
		\caption{A two-link network}
		\label{fig:n2}
	\end{figure}
	
	Suppose we observe $50$ cars passing link $1$ and $150$ vehicles (trucks and cars) passing link $2$, then we have
	
	\begin{eqnarray}
	y_1 &=& 50\nonumber\\
	y_2&=& 150\nonumber\\
	L_{a=1, i=1}^{b=1, h_2 = 1} &=& 1\nonumber\\
	L_{a=1, i=2}^{b=1, h_2 = 1} &=& 0\nonumber\\
	L_{a=2, i=1}^{b=2, h_2 = 1} &=& 1\nonumber\\
	L_{a=2, i=2}^{b=2, h_2 = 1} &=& 1\nonumber
	\end{eqnarray}
	
	Therefore
	\begin{eqnarray}
	y_1 &=& L_{a=1, i=1}^{b=1, h_2 = 1} x_{a=1, i=1}^{h_2 =1} +L_{a=1, i=2}^{b=1, h_2 = 1} x_{a=1, i=2}^{h_2 =1}\nonumber\\
	&=& x_{a=1, i=1}^{h_2 =1}\nonumber\\
	y_2 &=& L_{a=2, i=1}^{b=2, h_2 = 1} x_{a=2, i=1}^{h_2 =1} + L_{a=2, i=2}^{b=2, h_2 = 1} x_{a=2, i=2}^{h_2 =1}\nonumber\\
	&=&  x_{a=2, i=1}^{h_2 =1} + x_{a=2, i=2}^{h_2 =1}\nonumber
	\end{eqnarray}
	
	Similarly, if we observe the travel time for traversing link $1$ and link $2$ is $100$ for cars, and the average travel time of link $2$ is $70$ for cars and trucks.
	\begin{eqnarray}
	z_1 &=& 100\nonumber\\
	z_2 &=& 70\nonumber\\
	M_{a=1, i= 1}^{e = 1, h_2 = 1} &=& 1\nonumber\\
	M_{a=2, i= 1}^{e = 1, h_2 = 1} &=& 1\nonumber\\
	M_{a=2, i= 1}^{e = 2, h_2 = 1} &=& 0.5 \nonumber\\
	M_{a=2, i= 2}^{e = 2, h_2 = 1} &=& 0.5\nonumber
	\end{eqnarray}
	
	Then we have
	\begin{eqnarray}
	z_1 &=& M_{a=1, i= 1}^{e = 1, h_2 = 1} t_{a =1, i =1}^{h_2 = 1} + M_{a=2, i= 1}^{e = 1, h_2 = 1} t_{a =2, i =1}^{h_2 = 1}\nonumber\\
	&=&  t_{a =1, i =1}^{h_2 = 1} + t_{a =2, i =1}^{h_2 = 1}\nonumber\\
	z_2 &=& M_{a=2, i= 1}^{e = 2, h_2 = 1}t_{a =2, i =1}^{h_2 = 1} +M_{a=2, i= 2}^{e = 2, h_2 = 1}t_{a =2, i =2}^{h_2 = 1} \nonumber\\
	&=& \frac{t_{a =2, i =1}^{h_2 = 1} + t_{a =2, i =2}^{h_2 = 1}}{2}\nonumber
	\end{eqnarray}

	We note the ``average travel time for cars and trucks''  means the ``average'' of car travel time and truck travel time, respectively. It is also possible to consider the weighted average of travel time by car and truck flow together, which will be discussed in section~\ref{sec:comp}.
	
\end{example}

\subsection{MCDODE formulation}

The formulation for multi-class dynamic origin-destination estimation (MCDODE) is presented in Equation~\ref{eq:dode}.

\begin{equation}
\label{eq:dode}
\begin{array}{rrcllll}
\vspace{5pt}
\displaystyle \min_{\{q_{rsi}^{h_1}\}_{i,r,s,h_1}} & \multicolumn{4}{l}{\displaystyle w_1 \sum_{b \in B}  \left( y_b' - \sum_{i \in D} \sum_{a \in A} \sum_{h_2 \in H}  L_{ai}^{bh_2}  \left(\sum_{rs \in K_q} \sum_{k \in K_{rs}} \sum_{h_1 \in H}   \rho_{rsi}^{ka}(h_1, h_2) p_{rsi}^{kh_1} q_{rsi}^{h_1}\right) \right)^2 } \\
~ & \multicolumn{4}{l}{\displaystyle + w_2\sum_{e \in E} \left(z_e' - \sum_{i \in D}\sum_{a \in A} \sum_{h_2 \in H} M_{ai}^{eh_2} t_{ai}^{h_2} \right)^2} \\
\textrm{s.t.} & \left(\{t_{ai}^{h_2}, c_{rsi}^{kh_1}, \mathbf{\rho}_{rsi}^{ka}(h_1, h_2) \} \right)_{r,s,i,k,a,h_1, h_2} &=&  \Lambda(\{f_{rsi}^{kh_1}\}_{i,r,s,k,h_1}) & \\
~ & f_{rsi}^{kh_1} &=&  p_{rsi}^{kh_1} q_{rsi}^{h_1} , \forall rs \in K_q, k\in K_{rs}, i\in D, h_1 \in H&\\
~ & p_{rsi}^{kh_1} &=& \Psi_{rsi}^k\left( \tilde{\vec{c}}_{rsi}^{k h_1}, \tilde{\vec{t}}_{ai}^{h_1} \right), \forall rs \in K_q, k\in K_{rs}, i\in D, h_1 \in H& \\
~ & q_{rsi}^{h_1} & \geq & 0 , \forall rs \in K_q, h_1 \in H, i\in D&
\end{array}
\end{equation}
where $y_b'$ and $z_e'$ are the observed flow and travel time data collected from real-world. The link/path travel time can be computed through tracking cumulative curves in the DNL model \citep{lu2013dynamic}, and the computation for DAR will be discussed in section~\ref{sec:dar}. The parameters $w_1$ and $w_2$ are the weight for each objective, respectively. 

Equation~\ref{eq:dode} is a bi-level optimization problem with upper level minimizing the $\ell^2$ norm between observed and estimated flow and travel time, and the lower level solves the traffic assignment problem denoted by function $\Psi_{rsi}^k(\cdot)$ and $\Lambda(\cdot)$. When $\Psi_{rsi}^k(\cdot, \cdot)$ represents the Dynamic User Equilibrium (DUE) conditions, Equation~\ref{eq:dode} is a Mathematical Program with Equilibrium Constraints (MPEC) problem. In contrast, when $\Psi_{rsi}^k(\cdot, \cdot)$ represents the Logit-model, Equation~\ref{eq:dode} can be formulated as either a bi-level optimization problem or a single-level non-linear optimization problem \citep{davis1994exact, ma2018generalized}.

Additional data sources, such as the historical OD demand and survey data, can also be used in the demand estimation. The DODE methods with these data have been extensively studied \citep{zhang2008estimating, wu2018hierarchical}. This paper focuses on the computational graph approach with observed traffic flow and travel time data, while other data can be potentially incorporated into the proposed framework.

The solution algorithm for the bi-level problem has been extensively studied over the past few decades. Various heuristic methods have been developed to solve the bi-level problem efficiently. For example, heuristic iterative algorithms between lower and upper problem are widely used in traffic applications \citep{yang1995heuristic}. Sensitivity analysis of the traffic assignment problem has been adopted to evaluate the gradients of bi-level problem.  \citet{lu2013dynamic} estimated the gradient of path flow using the dynamic network loading models, and \citet{balakrishna2008time} approximated the gradient of OD demand through stochastic perturbations. \citet{osorio2019dynamic} approximate the gradient by linearizing the dynamic traffic assignment models. All the studies mentioned above aim at finding the gradients of the OD demand for the bi-level formulation.  In this paper, we propose to evaluate the gradient of OD demand analytically through the forward-backward algorithm on a computational graph, and details will be described in section~\ref{sec:comp}.

\subsection{Vectorizing the MCDODE formulation}
\label{sec:vec}

First of all, the variables involved in the MCDODE formulation are vectorized, and the vectorization is performed for each vehicle class separately. We set $N = |H|$ and denote the total number path as $\Pi = \sum_{rs} |K_{rs}|$, $K=|K_q|$. The vectorized variables are presented in Table~\ref{tab:vec}.

\begin{table}[h!]
	\begin{center}
		\caption{MCDODE framework variable vectorization}
		\label{tab:vec}
		\begin{tabular}{p{3cm}ccccp{3cm}}
			\hline
			Variable & Scalar & Vector& Dimension & Type & Description\\
			\hline\hline \rule{0pt}{3ex}
			OD flow & $q_{rsi}^h$ & $\vec{q}_i$ &$\mathbb{R}^{N|K|}$ & Dense & $q_{rsi}^h$ is place at entry $(h-1)|K| + k$\\
			\hline \rule{0pt}{3ex}
			Path flow & $f_{rsi}^{kh}$ &$\mathbf{f}_i$ & $\mathbb{R}^{N\Pi}$ & Dense & $f_{rsi}^{kh}$ is placed at entry $(h-1)\Pi + k$\\
			\hline \rule{0pt}{3ex}
			Link flow & $x_{ai}^h$ & $\mathbf{x}_i$ &$\mathbb{R}^{N|A|}$ & Dense & $x_{ai}^h$ is placed at entry $(N-1)|A| + k$\\
			\hline \rule{0pt}{3ex}
			Link travel time & $t_{ai}^h$ & $\mathbf{t}_i$ &$\mathbb{R}^{N|A|}$ & Dense & $t_{ai}^h$ is placed at entry $(N-1)|A| + k$\\
			\hline \rule{0pt}{3ex}
			Path travel time & $c_{rsi}^{kh}$ & $\mathbf{c}_i$ &$\mathbb{R}^{N\Pi}$ & Dense & $c_{rsi}^{kh}$ is placed at entry $(N-1)\Pi + k$\\
			\hline \rule{0pt}{3ex}
			Observed flow & $y_b$ & $\mathbf{y}$ &$\mathbb{R}^{|B|}$ & Dense & $y_b$ is placed at entry $b$\\
			\hline \rule{0pt}{3ex}
			Observed travel time & $z_e$ & $\mathbf{z}$ &$\mathbb{R}^{|E|}$ & Dense & $z_e$ is placed at entry $e$\\
			\hline \rule{0pt}{3ex}
			DAR matrix & $\rho_{rsi}^{ka}(h_1, h_2)$ &$\mathbf{\rho}_i$ & $\mathbb{R}^{N|A| \times N\Pi}$ & Sparse & $\rho_{rsi}^{ka}(h_1, h_2)$ is placed at entry $[(h_2-1)|A| + a, (h_1-1)\Pi + k]$\\
			\hline \rule{0pt}{3ex}
			Route choice matrix & $p_{rsi}^{kh}$ &$\mathbf{p}_i$ & $\mathbb{R}^{N\Pi \times N|K|}$ & Sparse & $p_{rsi}^{kh}$ is placed at entry $[(h-1)|\Pi| + k, (h-1)|K| + rs]$\\
			\hline\rule{0pt}{3ex}
			Observation/link incidence matrix & $L_{ai}^{bh}$ & $\vec{L}_i$ & $\mathbb{R}^{|B| \times N||A|}$ & Sparse & $L_{ai}^{bh}$ is placed at entry $\left[b,  (h-1)|A| + a \right]$\\
			\hline\rule{0pt}{3ex}
			Link travel time portion matrix & $M_{ai}^{eh}$ & $\vec{M}_i$ & $\mathbb{R}^{|E| \times N||A|}$ & Sparse & $M_{ai}^{eh}$ is placed at entry $\left[e,  (h-1)|A| + a \right]$\\
			\hline
		\end{tabular}
	\end{center}
\end{table}

With the vectorized variables, Equations~\ref{eq:linkod}, \ref{eq:fo} and \ref{eq:tl} can be rewritten in Equation~\ref{eq:vec}.

\begin{equation}
\begin{array}{lllllll}
\vspace{5pt}
\label{eq:vec}
\vec{x}_i &=& \mathbf{\rho}_i \vec{p}_i \vec{q}_i, \forall i \in D\\
\vec{y} &=&\sum_{i \in D} \vec{L}_i \vec{x}_i\\
\vec{z} & =& \sum_{i \in D} M_i \vec{t}_i \end{array}
\end{equation}

Multiplications between sparse matrix and sparse matrix, sparse matrix and dense vector are very efficient, especially on multi-core CPUs and Graphics Processing Units (GPUs). Therefore, Equation~\ref{eq:vec} can be evaluated efficiently.
The MCDODE formulation in Equation~\ref{eq:dode} can be cast into the vectorized form presented in Equation~\ref{eq:dode2}.

\begin{equation}
\label{eq:dode2}
\begin{array}{rrcllll}
\vspace{5pt}
\displaystyle \min_{\{\vec{q}_i \}_{i}} & \multicolumn{4}{l}{\displaystyle  w_1\left(  \norm{ \vec{y}'- \sum_{i \in D} \vec{L}_i \vec{x}_i  }_2^2  \right)
	+ w_2 \left(\norm{\vec{z}' - \sum_{i \in D} \vec{M}_i \vec{t}_i }_2^2\right)  } \\
\textrm{s.t.} & \left \{\vec{t}_i, \vec{c}_i, \mathbf{\rho}_i \ \right\}_{i} &= & \Lambda(\{\vec{f}_i\}_i) & \\
~ & \vec{f}_i &=&  \vec{p}_i\vec{q}_i & \forall  i\in D&\\
~ & \vec{x}_i &=&  \mathbf{\rho}_i \vec{f}_i & \forall  i\in D&\\
~ & \vec{p}_i &= & \Psi_i \left( \{\vec{c}_i\}_i, \{\vec{t}_i \}_i  \right)& \forall i\in D \\
~ & \vec{q}_i & \geq & 0 & \forall  i\in D&
\end{array}
\end{equation}
where $\Psi_i \left( \{\vec{c}_i\}_i, \{\vec{t}_i \}_i  \right)$ is the vectorized route choice function for vehicle class $i$. We further substitute the path flow and link flow to the objective function, as presented in Equation~\ref{eq:dode3}.

\begin{equation}
\label{eq:dode3}
\begin{array}{rrcllll}
\vspace{5pt}
\displaystyle \min_{\{\vec{q}_i \}_{i}} & \multicolumn{4}{l}{\displaystyle  w_1\left(  \norm{ \vec{y}'- \sum_{i \in D} \vec{L}_i \mathbf{\rho}_i \vec{p}_i\vec{q}_i }_2^2  \right) + w_2\left( \norm{\vec{z}' - \sum_{i \in D} \vec{M}_i \vec{t}_i }_2^2 \right)} \\
\textrm{s.t.} & \left \{\vec{t}_i, \vec{c}_i, \mathbf{\rho}_i \ \right\}_{i} &= & \Lambda( \{ \vec{f}_i \}_i) & \\
~ & \vec{p}_i &= &\Psi_i \left( \{\vec{c}_i\}_i, \{\vec{t}_i \}_i  \right)& \forall i\in D \\
~ & \vec{q}_i & \geq & 0 & \forall  i\in D&
\end{array}
\end{equation}

\subsection{A computational graph for MCDODE}
\label{sec:comp}

In order to solve the MCDODE problem, our goal is to obtain the gradient of OD demand for formulation~\ref{eq:dode3}. In this section, we propose a novel approach to obtain the gradient of OD demand through the forward-backward algorithm on a computational graph. First we cast equation~\ref{eq:dode3} into a computational graph representation, and Figure~\ref{fig:fb} describes the structure of the computational graph for MCDODE. The forward-backward algorithm runs on the computational graph, and algorithm consists of two processes: the forward iteration and the backward iteration.

In general, the forward iteration assumes that OD demand is fixed and it solves for the network conditions, while the backward iteration assumes the network conditions are fixed and it updates the OD demand. The whole process of forward-backward algorithm resembles some heuristic methods that solve the upper level and lower level problem iteratively \citep{yang1995heuristic}. However, the proposed algorithm can evaluate the instantaneous gradient of OD demand analytically in the backward iteration.

{\em Forward iteration:} Forward iteration takes multi-class OD demand as the input and solves the traffic assignment problem presented in Equation~\ref{eq:dta}.
\begin{equation}
\label{eq:dta}
\begin{array}{llllll}
\vspace{5pt}
\vec{f}_i &=& \vec{p}_i\vec{q}_i&\forall i \in D\\
\left \{ \vec{t}_i, \vec{c}_i, \mathbf{\rho}_i \ \right\}_{i} &= & \Lambda(\{ \vec{f}_i \}_i) &\\
\vec{p}_i &= &\Psi_i \left( \{\vec{c}_i\}_i, \{\vec{t}_i \}_i  \right)& \forall i\in D \\
\end{array}
\end{equation}

The forward iteration includes the multi-class dynamic network loading models $\Lambda$ and route choice models $\{ \Psi_i \}_i$. The output of the forward iteration is route choice portion $\vec{p}_i$, DAR matrix $\mathbf{\rho}_i$ and link/path travel time $(\vec{t}_i, \vec{c}_i)$. We omit the solution method for traffic assignment problem, as it has been extensively studied in many literature \citep{peeta2001foundations}.  After solving the dynamic traffic assignment models, the forward iteration compute the objective function (loss) in formulation~\ref{eq:dode3}, which can be represented by a series of equations in Equation~\ref{eq:loss} and \ref{eq:forward}.
\begin{equation}
\label{eq:loss}
\L = w_1 \L_1 + w_2 \L_2
\end{equation}

\begin{equation}
\begin{array}{lllllll}
\label{eq:forward}
\L_1 &=& \norm{\vec{y}'- \vec{y}}_2^2 && \L_2 &=& \norm{\vec{z}'- \vec{z}}_2^2\\
\vec{y} &=& \sum_{i \in D} \vec{L}_i \vec{x}_i&\quad&   \vec{z} &=&\sum_{i \in D} \vec{M}_i \vec{t}_i \\
\vec{x}_i &=& \mathbf{\rho}_i  \vec{f}_i  &\quad& \{\vec{t}_i\}_i &=& \bar{\Lambda}(\{\vec{x}_i\}_i)\\
\vec{f}_i &=& \vec{p}_i \vec{q}_i &\quad&\vec{x}_i &=& \mathbf{\rho}_i  \vec{f}_i\\
&~& &\quad& \vec{f}_i &= &\vec{p}_i \vec{q}_i
\end{array}
\end{equation}
where $\vec{y}'$ is the observed flow and $\vec{y}$ is the reproduction of the observed flow estimated from the traffic assignment model. Similarly, $\vec{z}'$ is the observed travel time and $\vec{z}$ is the reproduction of the observed travel time estimated from the traffic assignment model. We use $\bar{\Lambda}$ to represent a part of the function $\Lambda$, and $\bar{\Lambda}$ takes dynamic link flow as input and outputs the link travel time $\{\vec{t}_i\}_i$. Precisely, $\bar{\Lambda}$ represents the dynamic link models \citep{zhang2013modelling, jin2012link}.

{\em Backward iteration:} The backward iteration searches for the gradient of OD demand for formulation~\ref{eq:back} with the route choice portion $\vec{p}_i$, DAR matrix $\mathbf{\rho}_i$ and travel time $\{\vec{c}_i, \vec{t}_i\}_i$ known from the forward iteration.
\begin{equation}
\label{eq:back}
\begin{array}{rrcllll}
\vspace{5pt}
\displaystyle \min_{\{\vec{q}_i \}_{i}} & \multicolumn{4}{l}{\displaystyle w_1 \left(  \norm{ \vec{y}' - \sum_{i \in D} \vec{L}_i \mathbf{\rho}_i \vec{p}_i\vec{q}_i }_2^2  \right) + w_2\left( \norm{\vec{z}' - \sum_{i \in D} \vec{M}_i \vec{t}_i }_2^2\right)} \\
~ & \vec{q}_i & \geq & 0 & \forall  i\in D&
\end{array}
\end{equation}


When the gradient of OD demand is known, a projected gradient descent method can be used to solve Equation~\ref{eq:back}. The reason we call the solution process for Equation~\ref{eq:back} a ``backward iteration'' is that, the gradient of OD demand for Equation~\ref{eq:back} can be evaluated through the backpropagation (BP) method. Taking the derivative of the objective function step by step, we have a series of equations presented in Equation~\ref{eq:backward}

\begin{equation}
\begin{array}{llllllll}
\label{eq:backward}
\frac{\partial \L}{\partial \L_1} &=& w_1 &\quad& \frac{\partial \L}{\partial \L_2} &=& w_2\\\\
\frac{\partial \L_1}{\partial \vec{y}} &=& 2 \left( \vec{y}' - \sum_{i' \in D} \vec{L}_{i'} \mathbf{\rho}_{i'} \vec{p}_{i'}\vec{q}_{i'} \right) & \quad &\frac{\partial \L_2}{\partial \vec{z}} &=&2\left( \vec{z}' - \sum_{i' \in D} \vec{M}_{i'} \vec{t}_{i'}  \right)\\\\
\frac{\partial \L_1}{\partial \vec{x}_i} &=& -\vec{L}_{i}^T \frac{\partial \L_1}{\partial \vec{y}}  &\quad& \frac{\partial \L_2}{\partial \vec{t}_i} &=& -\vec{M}_i^T \frac{\partial \L_2}{\partial \vec{z}} \\\\
\frac{\partial \L_1}{\partial \vec{f}_i} &=& \mathbf{\rho}_{i}^T \frac{\partial \L_1}{\partial \vec{x}_i} &\quad& \frac{\partial \L_2}{\partial \vec{x}_i} &=& \frac{\partial\bar{\Lambda}(\{\vec{x}_i\}_i)}{\partial \vec{x}_i}   \frac{\partial \L_2}{\partial \vec{t}_i}\\\\
\frac{\partial \L_1}{\partial \vec{q}_i} &=& \vec{p}_{i}^T \frac{\partial \L_1}{\partial \vec{f}_i} &\quad & \frac{\partial \L_2}{\partial \vec{f}_i} &=& \mathbf{\rho}_{i}^T  \frac{\partial \L_2}{\partial \vec{x}_i}\\\\
&& & \quad &\frac{\partial \L_2}{\partial \vec{q}_i} &=& \vec{p}_{i}^T\frac{\partial \L_2}{\partial \vec{f}_i}
\end{array}
\end{equation}
Combining the equations in \ref{eq:backward}, the gradient of OD demand is presented in Equation~\ref{eq:grad}.
\begin{equation}
\begin{array}{llllllll}
\label{eq:grad}
\frac{\partial \L}{\partial \vec{q}_i} &=& \frac{\partial \L}{\partial \L_1} \frac{\partial \L_1}{\partial \vec{q}_i} + \frac{\partial \L}{\partial \L_2}\frac{\partial \L_2}{\partial \vec{q}_i}\\
&=&- 2 w_1 \vec{p}_{i}^T  \mathbf{\rho}_{i}^T \vec{L}_{i}^T \left( \vec{y}'- \sum_{i' \in D} \vec{L}_{i'} \mathbf{\rho}_{i'} \vec{p}_{i'}\vec{q}_{i'} \right) - 2w_2 \vec{p}_{i}^T  \mathbf{\rho}_{i}^T\frac{\partial\bar{\Lambda}(\{\vec{x}_i\}_i)}{\partial \vec{x}_i} \vec{M}_i^T  \left( \vec{z}' - \sum_{i' \in D} \vec{M}_{i'} \vec{t}_{'}  \right)
\end{array}
\end{equation}


The forward-backward algorithm is illustrated in Figure~\ref{fig:fb}. The solid arrow represents the forward iteration and the dashed arrow represents the backward iteration. We omit the forward and backward iterations for historical OD data, while the other processes are described above. The forward and backward iterations for historical OD data are straightforward by adding $w_3\L_3 = w_3 \sum_{i\in D}\norm{\vec{q}'_i - \vec{q}_i}^2_2$ to the MCDODE formulation~\ref{eq:dode3} \citep{zhang2008estimating}.

\begin{figure}[h]
	\centering
	\includegraphics[width=0.85\linewidth]{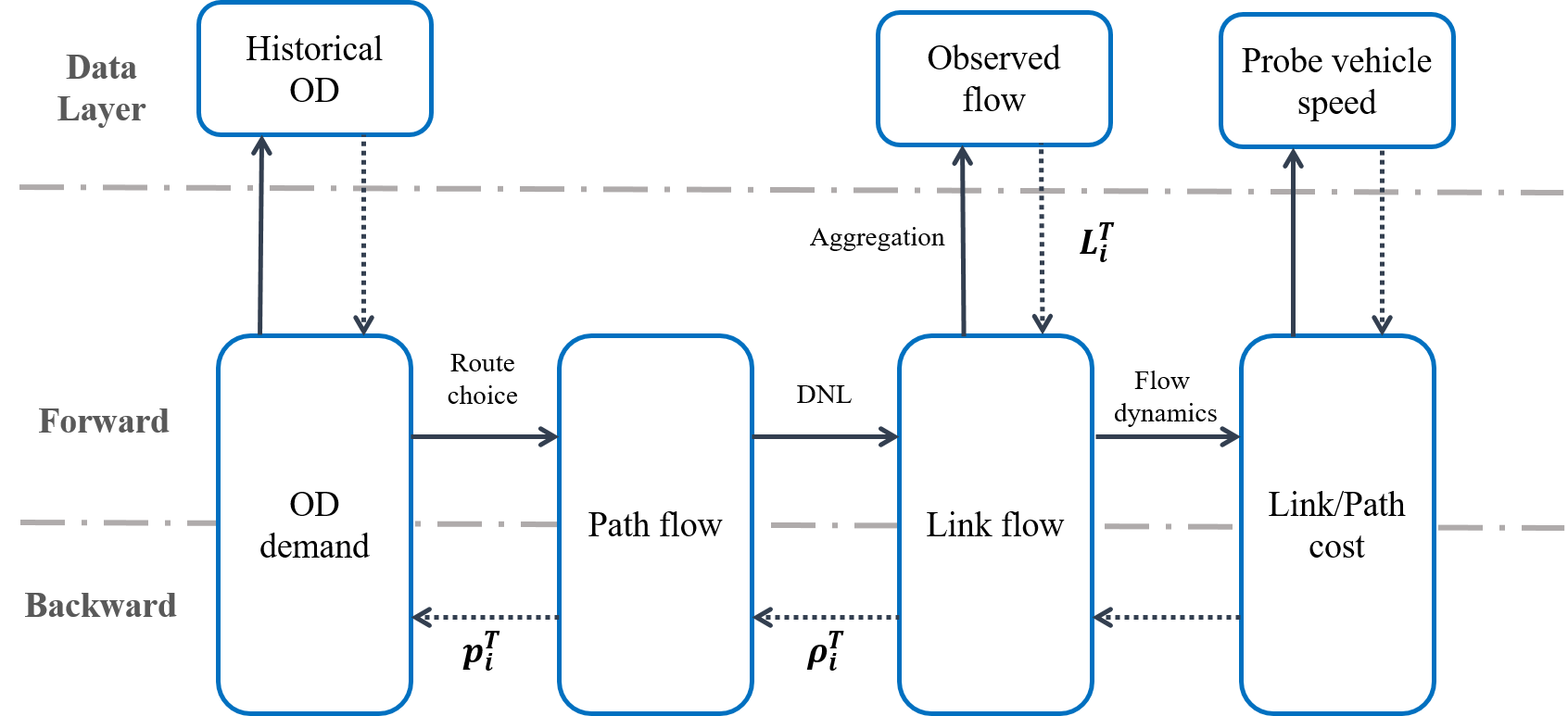}
	\caption{An illustration of the forward-backward algorithm.}
	\label{fig:fb}
\end{figure}

To solve the MCDODE formulation in Equation~\ref{eq:dode2}, we run the forward-backward algorithm to compute the gradient of OD demand. Then the projected gradient descent method is used to update the OD demand. There is a family of gradient-based methods that can be used, and we will discuss them in section~\ref{sec:sgd}. Comparisons among different methods will also be conducted in the numerical experiments.

Our computational graph approach shares many similarities with the deep learning models: 1) both models contain high dimensional parameters; 2) multi-core CPU and GPU can be used to speed up the solution process; 3) many advanced variants of gradient-based method can be used to solve the models; 4) Backpropagation method can be used to evaluate the gradient layer by layer \citep{rumelhart1985learning, wu2018hierarchical}. Potentially, all the techniques used in the training for deep learning can be used for the proposed computational graph. In this paper, we will test the advanced gradient-based method and multi-processing, while leaving other techniques, such as dropout \citep{gal2016dropout}, transfer learning \citep{pan2010survey} and regularization \citep{neyshabur2014search} for future research.

There are many flexible ways to incorporate observed data in the backward-forward algorithm. For example, it is possible to view the vehicle trajectories as samples from path flow and incorporate the trajectory data into the computational graph.
We can also compute the weighted average speed for cars and trucks to reproduce the ``average link travel time'' since the network conditions (flow, travel time) are fixed in the backward iterations. Hence the question left in Example~\ref{ex:1} is answered.

As for the stopping criteria, we first claim Proposition~\ref{prop:stop} holds by the definition of the forward-backward algorithm.
\begin{proposition}
	\label{prop:stop}
	In the proposed forward-backward algorithm, the DAR matrix, route choice portions and link/path travel time do not change if and only if $\frac{\partial \L}{\partial \vec{q}_i} = 0, \forall i$ during the forward and backward iterations.
\end{proposition}

Proposition~\ref{prop:stop} indicates that the forward-backward algorithm converges when either of the following two conditions hold: 1) in the forward iteration, the DAR matrix, route choice portions and link/path travel time do not change; 2) in the backward iteration, $\frac{\partial \L}{\partial \vec{q}_i} = 0$. During the forward-backward algorithm, we can either monitor the change of DAR matrix, route choice portions and link/path or the change of the gradient of OD demand. We can also see that the forward-backward algorithm converges to the local minimum for formulation~\ref{eq:dode3} in Proposition~\ref{prop:correct}.

\begin{proposition}
	\label{prop:correct}
	When the forward-backward algorithm converges, the estimated OD demand $\{\vec{q}_i\}_i$ is a local minimum for formulation~\ref{eq:dode3}.
\end{proposition}

\begin{proof}
	When the forward-backward algorithm converges, we know $\frac{\partial \L}{\partial \vec{q}_i} = 0$ by Proposition~\ref{prop:stop}. Since formulation~\ref{eq:back} is convex, $\frac{\partial^2 \L}{\partial \vec{q}_i^2} \succeq 0$. Therefore, $\{\vec{q}_i\}_i$ is the local minimum solution.
\end{proof}


\section{Solution Algorithms}
\label{sec:solution}
In this section, we first discuss several practical issues to complete the MCDODE framework with the forward-backward algorithm. We develop a multi-class traffic simulation package called MAC-POSTS for the network loading function $\Lambda(\cdot)$, and a Growing Tree algorithm is proposed to obtain the DAR matrix. Secondly, we discuss how to evaluate the derivative of link travel time in dynamic networks. Thirdly, we present how to incorporate multi-day observation data in the proposed MCDODE framework with multiprocessing. Lastly, the whole framework for MCDODE is presented.

\subsection{Multi-class traffic simulation}
\label{sec:simu}
The dynamic network loading function $\Lambda(\cdot)$ is fulfilled with the mesoscopic multi-class traffic simulation package Mobility Data Analytics Center - Prediction, Optimization, and Simulation toolkit for Transportation Systems (MAC-POSTS) developed by the Mobility Data Analytics Center at Carnegie Mellon University. To simulate heterogeneous traffic with multi-class vehicles like cars and trucks, the simulation package captures flow dynamics and outputs the traffic metrics for multi-class vehicles. The flow metrics include the traffic volumes, traffic speed and travel time. Due to the page limitation,  more details about the mesoscopic multi-class traffic simulation model in MAC-POSTS is presented in \citep{qian2017modeling, CARTRUCK}.

With the multi-class dynamic OD demand known, the mesoscopic multi-class traffic simulation model performs the following steps orderly in every loading interval ({\em e.g.} five seconds):
\begin{enumerate}
	\item Vehicle generation: multi-class vehicles are generated at origins according to the traffic demand.
	\item Routing: the route choice behaviors of all vehicles are updated, according to the route choice models.
	\item Node evolution: cars and trucks are moved through intersections following the intersection flow model.
	\item Link evolution: cars and trucks are moved on links following the link flow model.
	\item Network flow statistics: the model records link flow counts, link speeds, travel time and other network performance statistics.
\end{enumerate}

We note the loading interval is different from the interval defined in this paper, since loading interval is usually much shorter. After running the simulation, the route choice portion $\vec{p}_i$ and link/route travel time $(\vec{t}_i, \vec{c}_i)$ for each vehicle class can be obtained based on the simulation results. The DAR matrix $\mathbf{\rho}_i$ for each vehicle class can also be obtained by constructing the tree-based cumulative curves during the simulation process, and details are presented in section \ref{sec:dar}.

\subsection{Tree-based cumulative curve}
\label{sec:dar}
In this section, we develop a novel method to compute the DAR matrix during the traffic simulation. Computing the DAR matrix during the simulation is more efficient than obtaining the DAR matrix after the simulation. However, computing the dynamic assignment ratio (DAR) during the traffic simulation is challenging. A naive method to obtain the DAR matrix is by recording the trajectories of all simulated vehicles in the DNL process. This method iterates across all paths and links over all time intervals, and in each iteration the method computes the number of vehicles arriving at a specific link from a specific path. Since the dimension of DAR matrix increases exponentially with respect to the size of network and the number of time intervals, the naive method is computational implausible for large-scale networks. In this section, we propose a novel method to compute the DAR matrix through the tree-based cumulative curves, and the proposed method is efficient in both computational time (time complexity) and memory (space complexity).

We define $\chi_{arsi}^{kh_1}(\cdot)$ as the tree-based cumulative curve of link $a$ for class-$i$ vehicles departing from path $k$ in OD pair $rs$ in time interval $h_1$.  $\chi_{arsi}^{kh_1}(t)$ takes the time $t$ as input and outputs the total number of vehicles departing in time interval $h_1$ from path $rsk$ and arriving at link $a$ before time $t$. Then the DAR can be computed by Equation~\ref{eq:dar2}.
\begin{equation}
\label{eq:dar2}
\rho_{rsi}^{ka}(h_1, h_2) = \frac{\chi_{arsi}^{kh_1}(\overline{t_2}) - \chi_{arsi}^{kh_1}(\underline{t_2})}{f_{rsi}^{kh_1}}
\end{equation}
where $\underline{t_2}$ is the beginning of time interval $h_2$ and $\overline{t_2}$ is the end of time interval $h_2$.

We note that $\chi_{arsi}^{kh_1}(\cdot)$ records the cumulative curves for each path flow and departing time interval separately, and hence it requires more memory and computational power than the standard link-based cumulative curve \citep{lu2013dynamic}. However, only a very small fraction of vehicles pass a specific link $a$ during the simulation. There are only a small fraction of paths containing a specific link, and hence the $\chi_{arsi}^{kh_1}(\cdot)$ is sparse in terms of path indices $rsk$ and time indices $h_1$. Using this intuition, we develop a Growing Tree algorithm to build the tree-based cumulative curve $\chi_{arsi}^{kh_1}(\cdot)$ for each link $a$. Since the algorithm is tree-based, so we refer $\chi_{arsi}^{kh_1}(\cdot)$ as the tree-based cumulative curves.

The process of the Growing Tree algorithm is presented in Algorithm~\ref{alg:gt}.

\begin{algorithm}[H]
	\SetKwInOut{Input}{Input}
	\SetKwInOut{Output}{Output}
	\underline{\texttt{GrowingTree}} $\left(S, n\right)$\;
	\Input{Traffic simulator $S$, number of loading intervals $n$}
	\Output{Tree-based cumulative curves $\chi$}
	Initialize an empty dictionary $\chi$\;
	\For{$(i=0;~i<n;~++i)$ }{
		Run the simulator $S$ for one loading interval\;
		\For{$a \in A$}{
			Initialize an empty dictionary $\chi[a]$\;
			Extract the set of vehicles going in link $a$\, and denote it as $Q$\;
			\For{$v \in Q$}{
				Suppose vehicle $v$ follows path $k$ in OD pair $rs$ and departs in time interval $h_1$ and the vehicle class is $i$\;
				\If{$i$ is not the key of dictionary $\chi[a]$}{
					Initialize $\chi[a][i]$ with an empty dictionary\;
				}		
				\If{$h_1$ is not the key of dictionary $\chi[a][i]$}{
					Initialize $\chi[a][i][h_1]$ with an empty dictionary\;
				}
				\If {$rsk$ is not the key of dictionary $\chi[a][i][h_1]$}{
					Initialize $\chi[a][i][h_1][rsk]$ with an empty cumulative curve\;
				}
				Add record $(i, 1)$ to the cumulative curve $\chi[a][i][h_1][rsk]$\;
			}
		}
	}
	
	\caption{Growing Tree algorithm for constructing $\chi_{arsi}^{kh_1}(\cdot)$}
	\label{alg:gt}
\end{algorithm}
In the algorithm, a record $(i, 1)$ means one vehicle arriving at the link $a$ at time $t$. $\chi_{arsi}^{kh_1}(\cdot)$ is constructed as a tree, as presented in Figure~\ref{fig:gt}. During the construction, when a vehicle transverses a link, a leaf containing a standard cumulative curve is either created or updated to record the location of that vehicle. The tree $\chi_{arsi}^{kh_1}(\cdot)$ is unbalanced, so a hashmap-based tree is more memory efficient. In Algorithm~\ref{alg:gt}, a dictionary refers to a key-value mapping implemented by hashmap, and readers can view the dictionary as one of the following data structures: \texttt{dictionary} in Python, \texttt{HashMap} in Java,  or \texttt{unordered\_map} in C++.

\begin{figure}[h]
	\centering
	\includegraphics[width=0.85\linewidth]{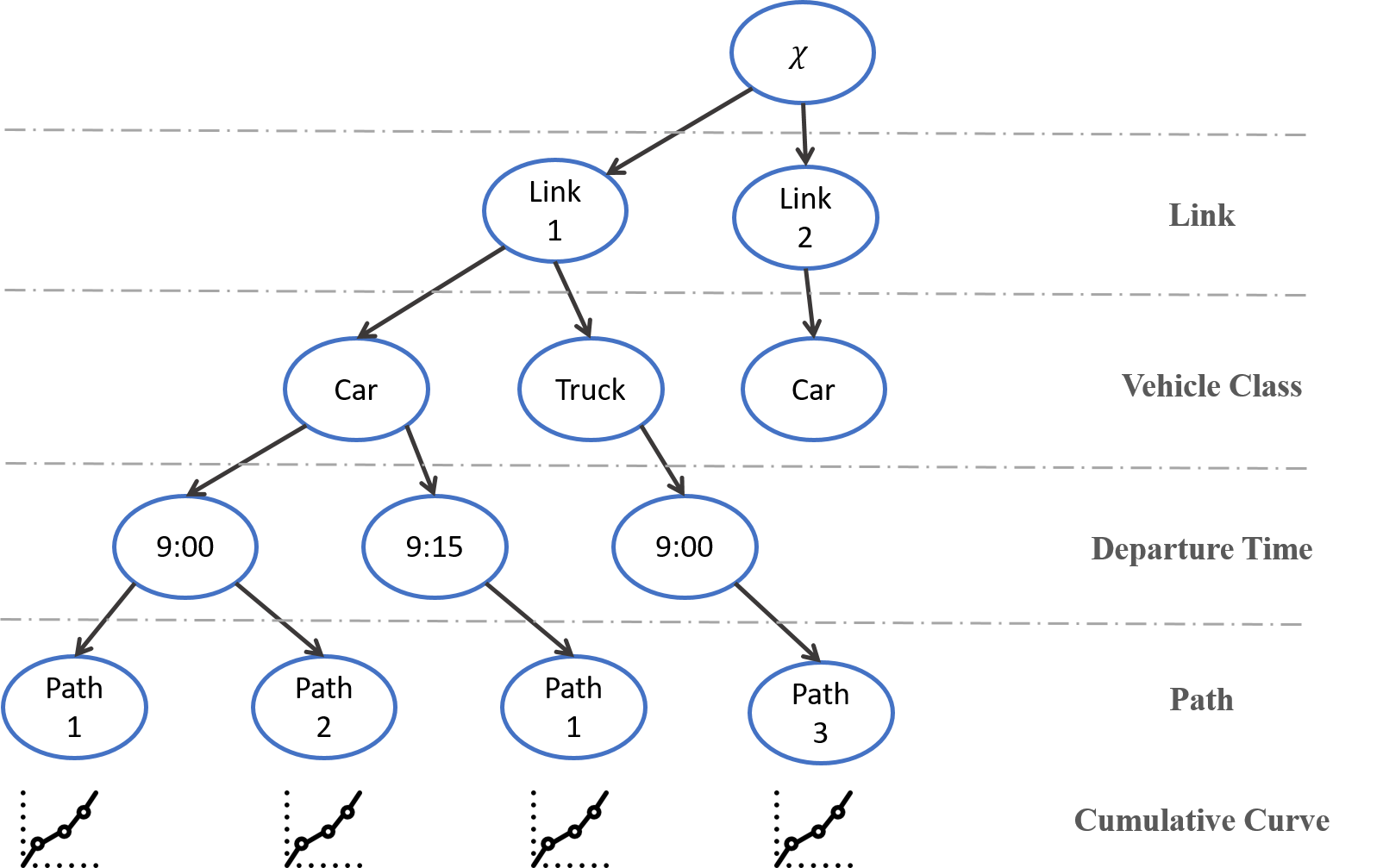}
	\caption{An illustration of the Growing Tree algorithm.}
	\label{fig:gt}
\end{figure}

\subsection{Derivatives of link travel time}
In the backward iteration described in section~\ref{sec:comp}, we need to compute the derivatives of link travel time $\frac{\partial\bar{\Lambda}(\{\vec{x}_i\}_i)}{\partial \vec{x}_i}$. The function $\bar{\Lambda}$ represents the link flow models which include, but is not limited to, point queue, spatial queue, cell transmission model, link transmission model and link queue model \citep{jin2012link, zhang2013modelling}. There is no closed-form for link flow models, hence the derivatives of function $\bar{\Lambda}$ can be challenging to evaluate analytically. In contrast, there exists methods to approximate the link travel time derivatives for the simulation-based models, and the basic idea of this kind of methods is to examine the extra link travel time induced by a marginal vehicle added to the link. Readers are refereed to \citet{qian2012system, lu2013dynamic} for the implementation details.
In this paper, we adopt the approximation approach discussed in \citet{lu2013dynamic} to evaluate $\frac{\partial\bar{\Lambda}(\{\vec{x}_i\}_i)}{\partial \vec{x}_i}$.

\subsection{Incorporating multi-day observations with multiprocessing}
\label{sec:sgd}
From formulation \ref{eq:dode3}, the multi-class demand is estimated with one data sample $(\vec{y}', \vec{z}')$. In real world applications, we may observe the flow and travel time on multiple days. Suppose we collect data samples for $M$ days and we let $(\vec{y}_m', \vec{z}_m')$ denote the observed flow and travel time on day $m$, then the MCDODE problem~\ref{eq:dode3} can be extended to accommodate multi-day observations in formulation~\ref{eq:dode4}

\begin{equation}
\label{eq:dode4}
\begin{array}{rrcllll}
\vspace{5pt}
\displaystyle \min_{\{\vec{q}_i \}_{i}} & \multicolumn{4}{l}{\displaystyle \frac{1}{M} \sum_{1 \leq m \leq M} \left[ w_1\left(  \norm{ \vec{y}_m'- \sum_{i \in D} \vec{L}_i \mathbf{\rho}_i \vec{p}_i\vec{q}_i }_2^2  \right) + w_2\left( \norm{\vec{z}_m' - \sum_{i \in D} \vec{M}_i \vec{t}_i }_2^2 \right)\right]} \\
\textrm{s.t.} & \left \{\vec{t}_i, \vec{c}_i, \mathbf{\rho}_i \ \right\}_{i} &= & \Lambda( \{ \vec{f}_i \}_i) & \\
~ & \vec{p}_i &= &\Psi_i \left( \{\vec{c}_i\}_i, \{\vec{t}_i \}_i  \right)& \forall i\in D \\
~ & \vec{q}_i & \geq & 0 & \forall  i\in D&
\end{array}
\end{equation}

Formulation~\ref{eq:dode4} can be directly solved using the forward-backward algorithm. We only need to compute the gradients of OD demand for each data sample and use the average gradient over all data samples to update the OD demand during the backward iteration. This process is the same as Gradient Descent (GD) method. In addition, the stochastic gradient descent (SGD) method can also be used to solve formulation~\ref{eq:dode4}. In the process of SGD, we evaluate the gradient of OD demand for one randomly selected data sample and then use it to update the OD demand. The comparisons between GD and SGD exist in many machine learning models, readers are refereed to \citep{saad1998online} for more details. Many advanced gradient descent methods can also be used to solve formulation~\ref{eq:dode4}. For example, Adagrad is one of the most representatives of variants of SGD with adaptive step sizes, and it is often used in the optimization of deep neural networks \citep{duchi2011adaptive}.

We can further speed up the solution process for formulation~\ref{eq:dode4} by utilizing the power of multiprocessing. The delayed stochastic gradient descent (delayed-SGD) method can evaluate the gradient of multiple data samples on a multi-core CPU at the same time \citep{zinkevich2009slow}. Each core is responsible for evaluating one single data sample at one time.  Comparing to the traditional SGD, the delayed-SGD makes full use of the multi-core CPU and hence it can solve the MCDODE framework more efficiently.  It is also possible to extend Formulation~\ref{eq:dode4} to incorporate multi-day data that are observed on different links. To do that, we can replace $\vec{L}_i$ and $\vec{M}_i$ to $\vec{L}_{im}$ and $\vec{M}_{im}$ for each day $m$ separately.

\subsection{Solution framework}
\label{sec:framework}
The solution algorithm for MCDODE is summarized in Table~\ref{tab:sol}.
\begin{table}[h]
	\begin{tabular}{p{3cm}p{11.8cm}}
		\textbf{Algorithm}& \textbf{[\textit{MCDODE-FRAMEWORK}]} \\\hline
		\textit{Step 0} & \textit{Initialization.} Initialize the OD demand vector $\{\mathbf{q}_i \}_i$ for each vehicle class. \\\hline
		\textit{Step 1} & \textit{Forward iteration.}  Solve the traffic assignment model presented in equation~\ref{eq:dta} with OD demand $\{\vec{q}_i\}_i$, and construct the tree-based cumulative curve $\chi$ through Growing Tree algorithm presented in Algorithm~\ref{alg:gt}. \\\hline
		\textit{Step 2} & \textit{Variable retrieval.} Extract the link/path travel time from the simulation model, compute the route choice matrix from route choice model by Equation~\ref{eq:gen_choice}, and obtain the DAR matrix from the tree-base cumulative curves by Equation~\ref{eq:dar2}. \\\hline
		\textit{Step 3} & \textit{Backward iteration.} Compute the gradient of OD demand using the backward iteration presented in Equation~\ref{eq:backward} and \ref{eq:grad}. \\\hline
		\textit{Step 4} & \textit{Update OD demand.}  Update the OD demand with the gradient-based projection method discussed in section~\ref{sec:sgd}.\\ \hline
		\textit{Step 5} & \textit{Convergence check.}  Stop when the change of OD demand $\{\mathbf{q}_i \}_i$ is less than tolerance. Otherwise, go to Step 1.\\\hline
	\end{tabular}
	\caption{MCDODE solution framework}
	\label{tab:sol}
\end{table}

\section{Numerical Experiments}
\label{sec:exp}
In this section, we first examine the proposed MCDODE framework in a small network. Estimation results are presented and discussed. We examine the effects of multiprocessing, compare different variants of gradient descent methods, and conduct the sensitivity analysis of the MCDODE framework. In addition, the effectiveness, efficiency and scalability of the MCDODE framework are demonstrated in a large-scale network. All the experiments in this section are conducted on a desktop with Intel Core i7-6700K CPU 4.00GHz $\times$ 8, 2133 MHz 2 $\times$ 16GB RAM, 500GB SSD.

\subsection{A small network}
We first work with a small network with seven links, three paths and one O-D pair, as presented in Figure~\ref{fig:snetwork}. Two classes of vehicles are considered: cars and trucks. Link 1 and Link 7 are OD connectors, and we use the identical triangular fundamental diagram (FD) for the rest of 5 links. In the FD, length of each road segment is $0.55$ mile, free flow speed is $35$ miles/hour for car and $25$ miles/hour for truck, flow capacity is 2,200 vehicles/hour for car and 1,200 vehicles/hour for truck, and the holding capacity is 200 vehicles/mile for car and 80 vehicles/mile for truck.

\begin{figure}[h]
	\centering
	\includegraphics[width=0.85\linewidth]{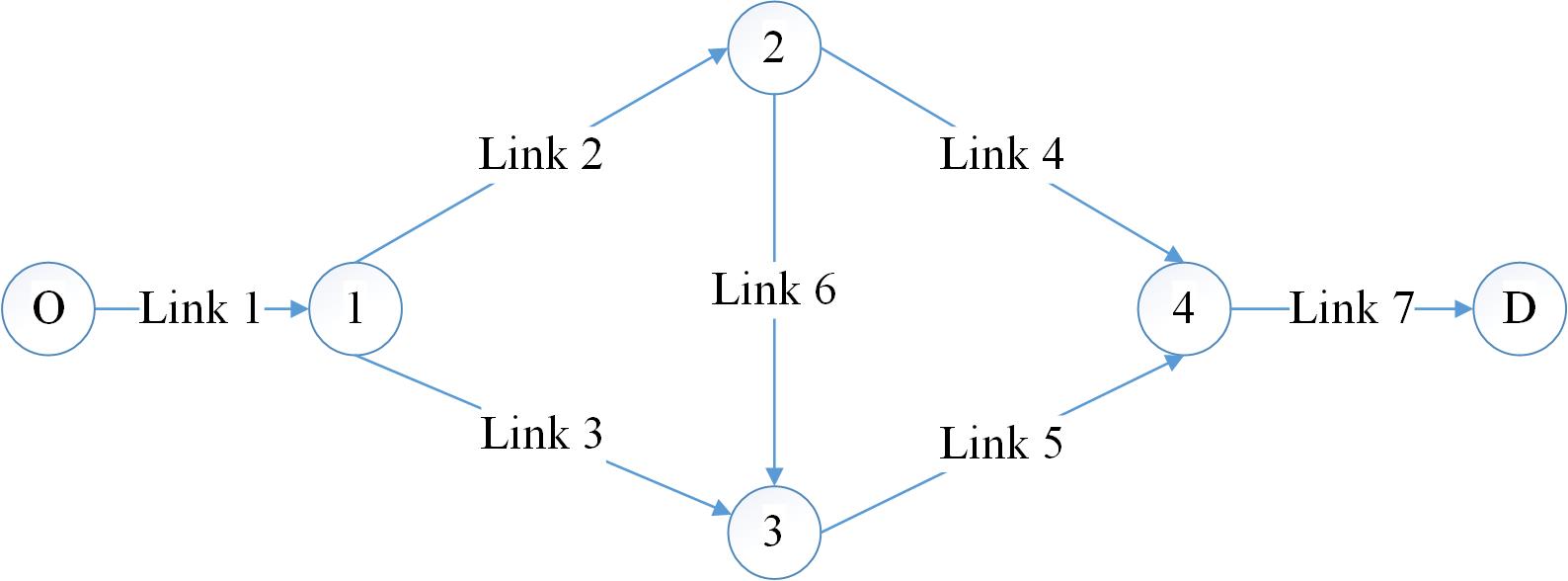}
	\caption{A small network.}
	\label{fig:snetwork}
\end{figure}

We generate the multi-class dynamic OD demand by random number generators and then treat it as the ``true'' OD demand. We generate the observed flow by running the multi-class network simulation model and then adding the noise. The performance of the MCDODE estimation formulation is assessed by comparing the estimated flow with the ``true'' flow (flow includes observed flow, link flow, OD demand) \citep{antoniou2015towards}. We use R-square between the ``true'' flow (travel time) and estimated flow (travel time) to measure the estimation accuracy.

{\bf Baseline setting:} in the small network, we randomly sample the ``true'' car and truck OD demand from uniform distributions $\texttt{Unif}(0, 300)$ and $\texttt{Unif}(0,60)$, respectively. We randomly generate the route choice portions and treat them as unknown, then we run the MAC-POSTS $\Lambda$ to obtain the ``true'' network conditions. We construct the observed flow as follows: firstly we randomly generate the matrix $\{\vec{L}_i\}_i$ by a Bernoulli distribution with $p = 0.5$ for link $3,4,5,6$ and leave the flow of link 2 hidden from all observations; secondly we aggregate the ``true'' link flow to obtain the observed flow by $\{\vec{L}_i\}_i$; thirdly we multiply $1+\varepsilon$ to the ``true'' observed flow to get the observed flow with noise, where $\varepsilon \sim \texttt{Unif}(-\xi, \xi)$ and $\xi \in [0,1)$ represents the noise level. We consider $10$ time intervals, and each time interval represents fifteen minutes. We set $|B| = 10$, and $6$ observations are from car flow and $4$ observations are from truck flow. Assuming we directly observe the travel time for link $3,4,5,6$ for cars and trucks separately, $\{ \vec{M}_i \}_i$ and $E$ can be constructed accordingly. We set $w_1= 1, w_2=0.01$. We also add noise to the observed link travel time using the same method as observed flow. We observe $8$ data samples $M=8$ and the noise level $\xi = 0.1$. We use single-process Adagrad with step size $1$ as the solution method, and the initial OD demand is generated from  $\texttt{Unif}(0, 15)$ and $\texttt{Unif}(0, 3)$ for car and trucks, respectively. The above setting is called the baseline setting.

\subsubsection{Basic estimation results}

In this section, we examine the basic estimation results of the proposed MCDODE framework for the baseline setting. We run the MCDODE framework presented in section~\ref{sec:framework} until convergence. The change of loss $\L$ against the number of iterations using Adagrad is presented in Figure~\ref{fig:loss}.

\begin{figure}[h]
	\centering
	\includegraphics[width=0.85\linewidth]{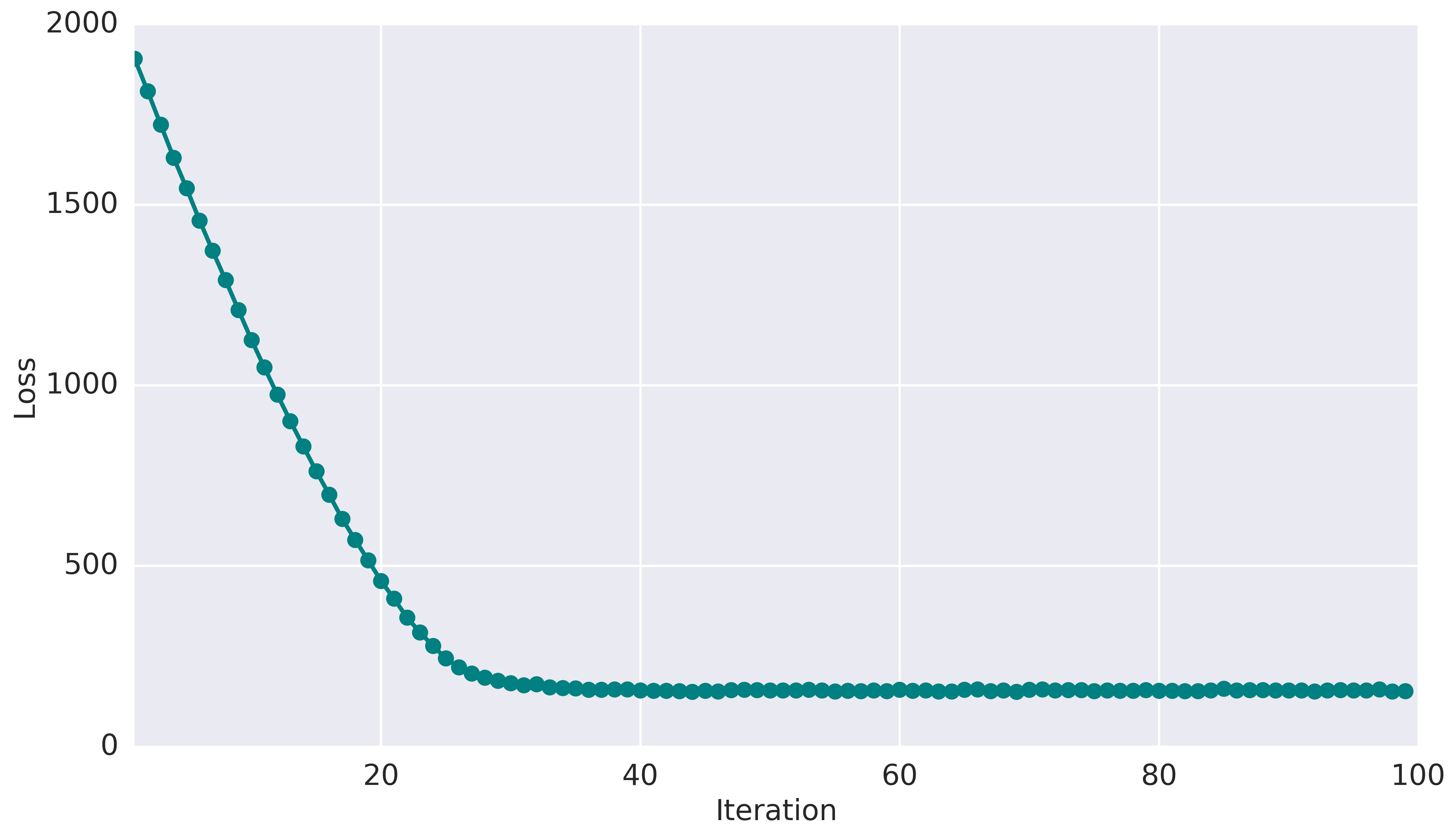}
	\caption{Convergence curve for the loss $\L$}
	\label{fig:loss}
\end{figure}

To analyze the convergence of the observed flow and travel time separately, we decompose the loss $\L$ into four components: car flow, car travel time, truck flow and truck travel time. We plot the loss for the four components separately, and the results are presented in Figure~\ref{fig:loss2}. Note we normalize the loss such that it is between $0$ and $1$.

\begin{figure}[h]
	\centering
	\includegraphics[width=0.85\linewidth]{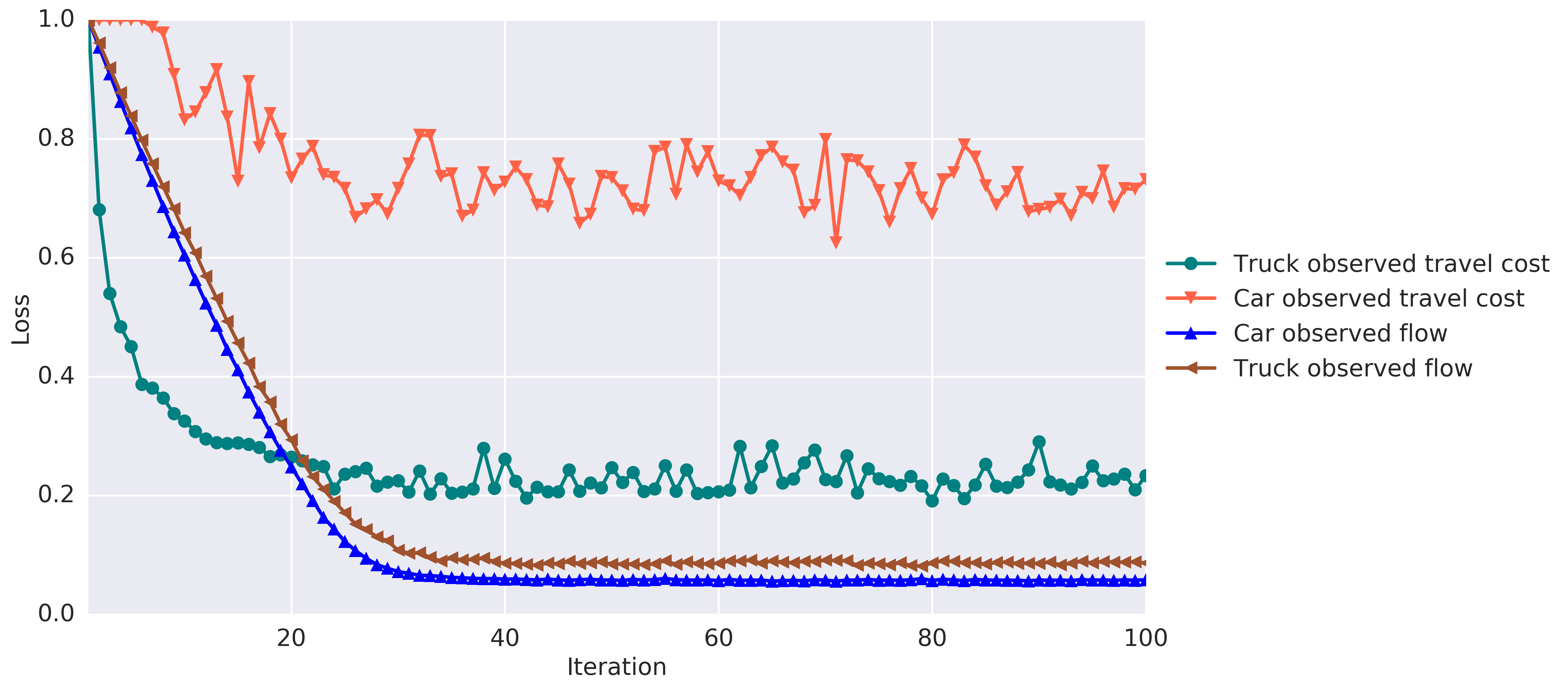}
	\caption{Decomposed convergence curve for observed flow and travel time (normalized)}
	\label{fig:loss2}
\end{figure}

The comparisons between the ``true'' and estimated values for observed flow, link flow, link travel time and OD demand are presented in Figure~\ref{fig:sy}, \ref{fig:sx}, \ref{fig:stt} and \ref{fig:sf}, respectively.

\begin{figure}[h]
	\centering
	\includegraphics[width=0.7\linewidth]{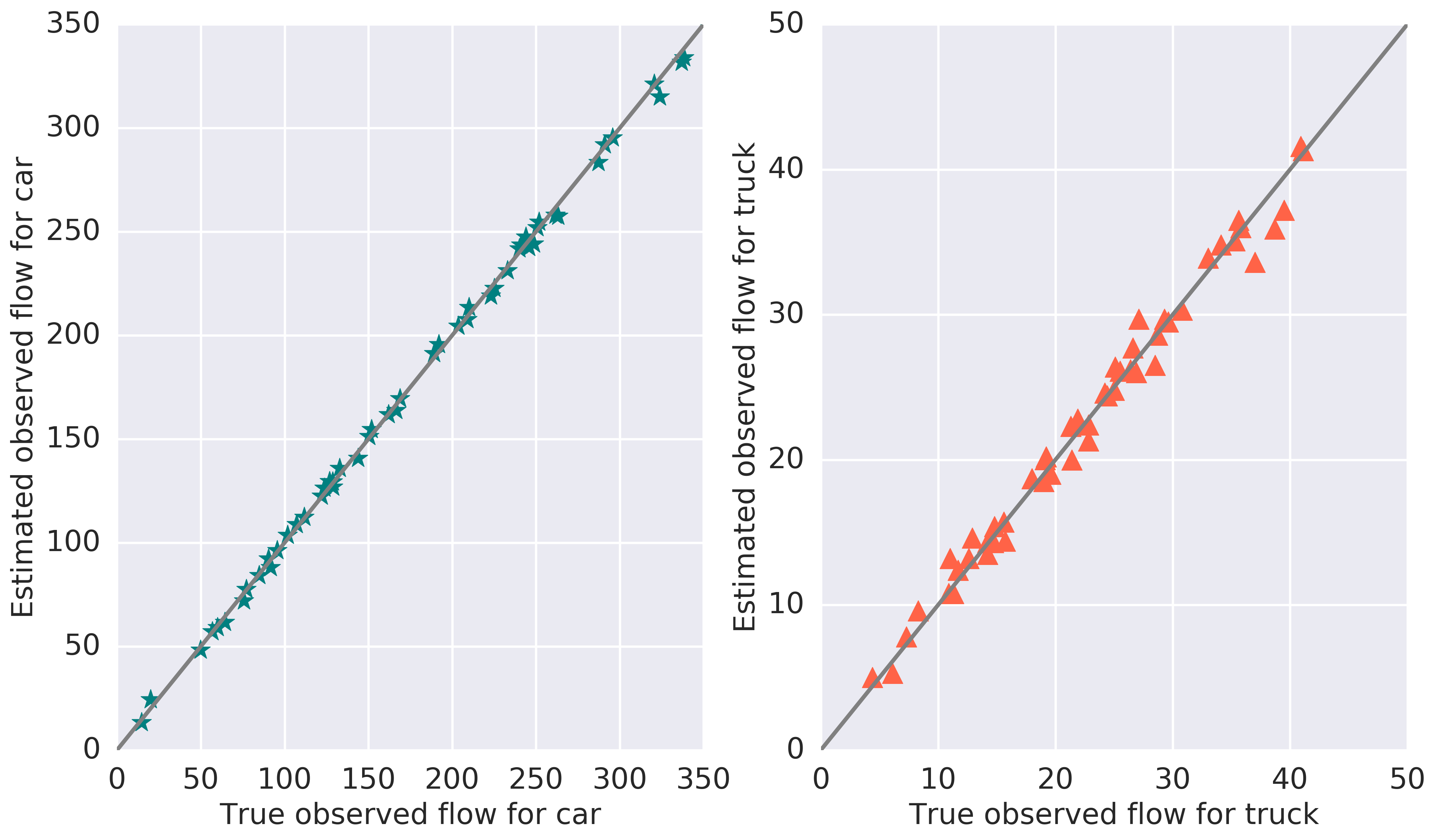}
	\caption{Estimated and ``true'' observed flow for cars and trucks (unit:vehicle/15mins).}
	\label{fig:sy}
\end{figure}

\begin{figure}[h]
	\centering
	\includegraphics[width=0.7\linewidth]{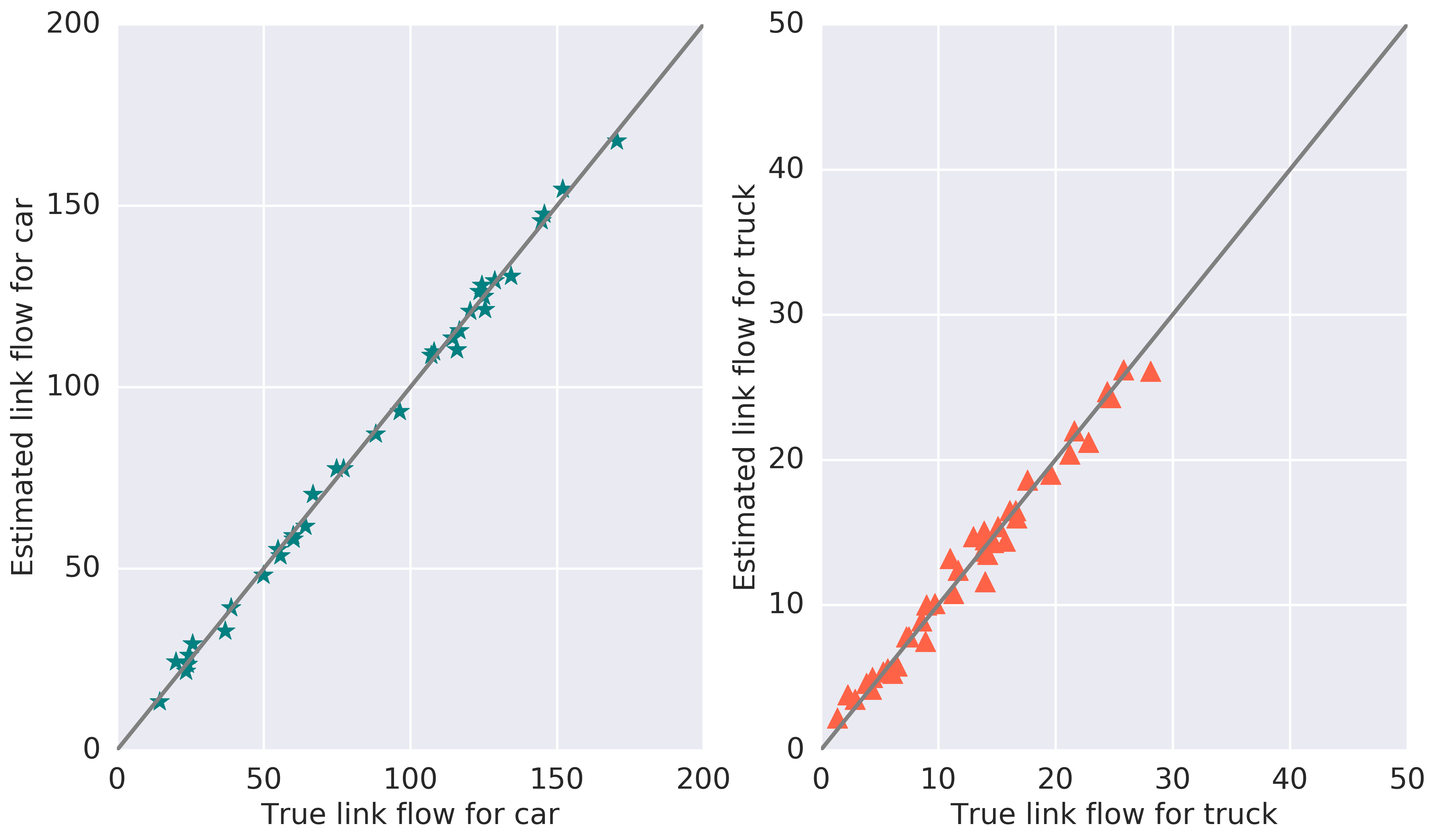}
	\caption{Estimated and ``true'' link flow for cars and trucks (unit:vehicle/15mins).}
	\label{fig:sx}
\end{figure}

\begin{figure}[h]
	\centering
	\includegraphics[width=0.7\linewidth]{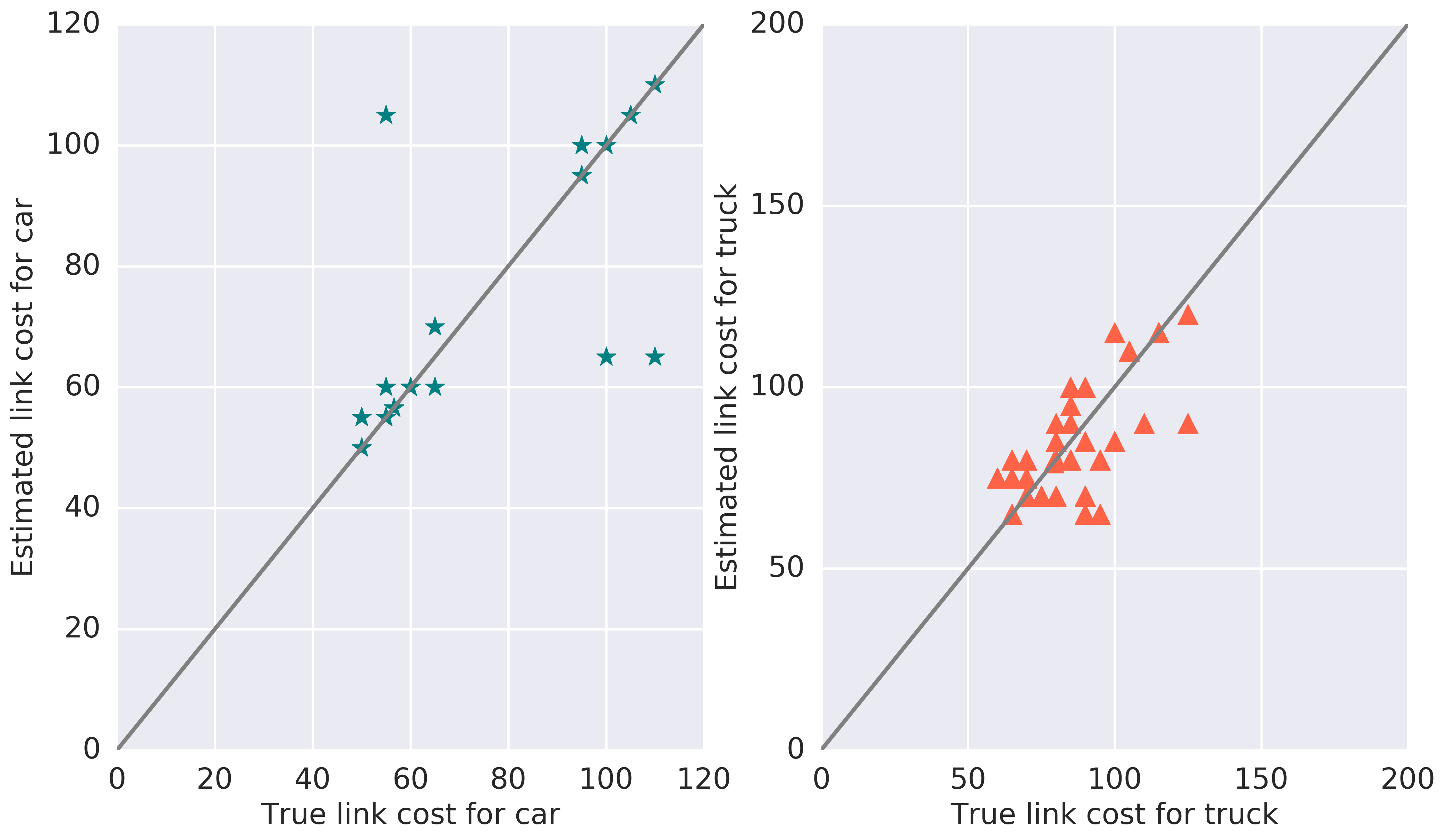}
	\caption{Estimated and ``true'' link travel time for cars and trucks (unit:seconds).}
	\label{fig:stt}
\end{figure}

\begin{figure}[h]
	\centering
	\includegraphics[width=0.7\linewidth]{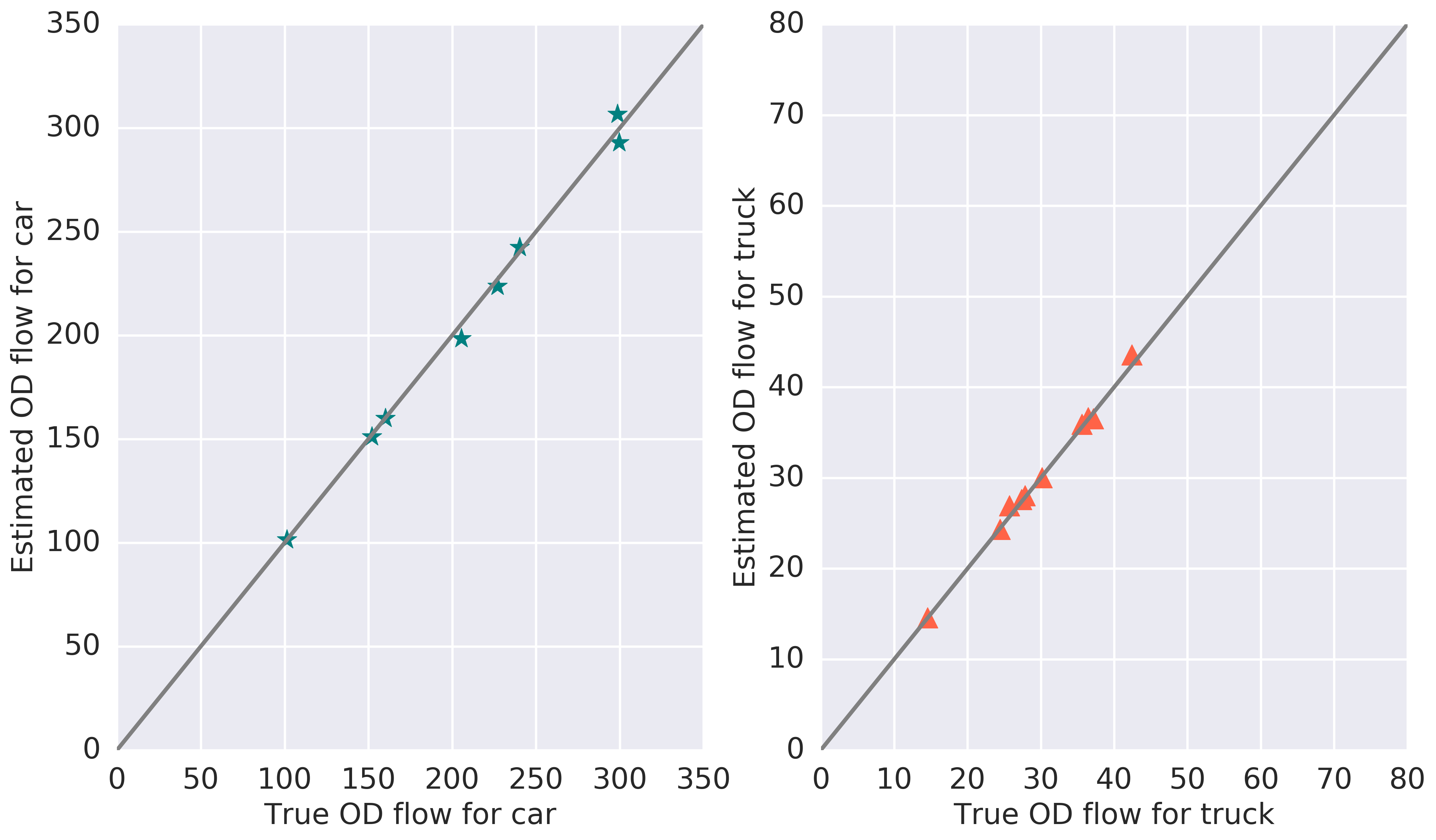}
	\caption{Estimated and ``true'' OD demand for cars and trucks (unit:vehicle/15mins).}
	\label{fig:sf}
\end{figure}

\begin{table}[h]
	\centering
	\begin{tabular}{c|cc}
		\hline
		~ & Car & Truck \\
		\hline\hline
		Observed flow & 0.9992 & 0.9858 \\
		Link flow & 0.9982 & 0.9808 \\
		Link travel time & 0.9309 & 0.9586 \\
		OD demand & 0.9965 & 0.9940 \\
		\hline
	\end{tabular}
	\caption{R-square between the ``true'' and estimated flow and travel time for cars and trucks}
	\label{tab:sr}
\end{table}

The  R-squares between the ``true'' and estimated flow (travel time) are presented in Table~\ref{tab:sr}. The proposed MCDODE framework yields accurate estimation of the multi-class dynamic OD demand in the small network. The average R-squares for car flow and truck flow are above 0.98. The estimation accuracy for truck is lower than car, which is probably attributed to the low truck demand. Since the multi-class traffic loading model is discretized and stochastic \citep{qian2017modeling}, low demand may incur a large variance in the simulation results. Therefore, the gradient of truck flow becomes noisy when the demand is low.

The R-square for the travel time is lower than that for flows, since the derivative of travel time  $\frac{\partial\bar{\Lambda}(\{\vec{x}_i\}_i)}{\partial \vec{x}_i}$ is approximated by the simulation rather than evaluated in a closed-form. Again, due to the discretization and stochasticity of the simulation model, the approximations of $\frac{\partial\bar{\Lambda}(\{\vec{x}_i\}_i)}{\partial \vec{x}_i}$ can be noisy.

\subsubsection{Comparing different gradient-based methods}
In this section, we examine the performance of three gradient-based methods: gradient descent (GD), stochastic gradient descent (SGD) method and Adagrad. We solve the MCDODE problem for the baseline setting three times using different gradient-based methods, and we plot the convergence curves for the three methods in Figure~\ref{fig:comp}.

\begin{figure}[h]
	\centering
	\includegraphics[width=0.85\linewidth]{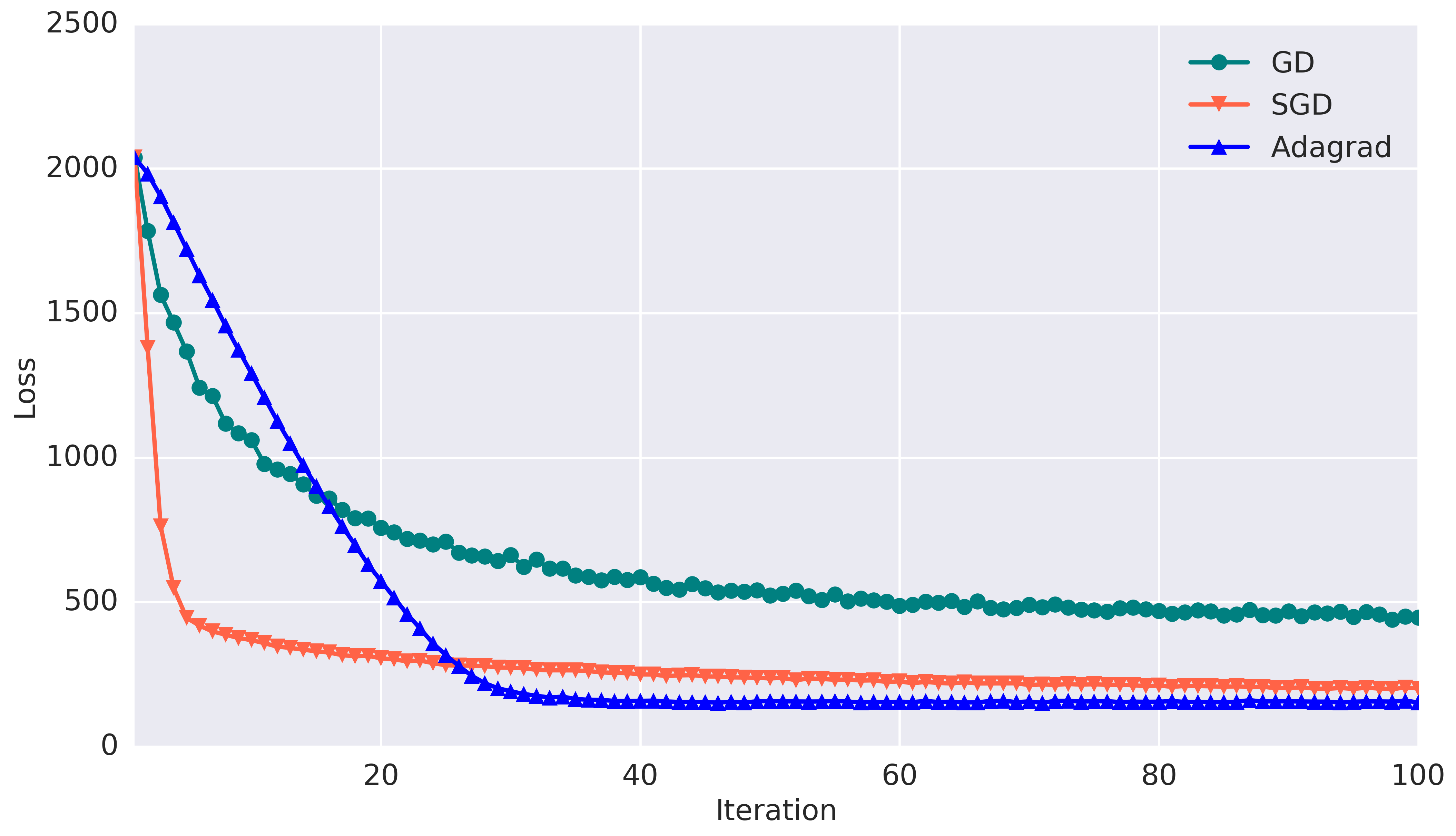}
	\caption{Comparison among GD, SGD and Adagrad}
	\label{fig:comp}
\end{figure}

For all three method, it takes less than $2$ minutes to complete the $100$ iterations. One can clearly see that  SGD outperforms the GD throughout the $100$ iterations. Though GD and SGD converge faster in the first $20$ iterations, the Adagrad outperforms both methods in terms of convergence rate and final loss after $100$ iterations. The reason for the best performance of Adagrad is probably because Adagrad can select the step size adaptively during the solution, hence it maintains a good convergence rate throughout the $100$ iterations. Adagrad is adopted as the standard method to solve the MCDODE formulation in the baseline setting.

\subsubsection{Multiprocessing}

In this section, we demonstrate the power of multiprocessing in solving the MCDODE framework. We solve the baseline setting four times using different number of processes. We use the delayed version of Adagrad to enable the multiprocessing \citep{zinkevich2009slow}.
We examine the convergence curves for 1, 2, 4, and 8 processes and plot the results in Figure~\ref{fig:mul}. We note that different from previous figures, the x-axis in Figure~\ref{fig:mul} is the time rather than iterations.
\begin{figure}[h]
	\centering
	\includegraphics[width=0.85\linewidth]{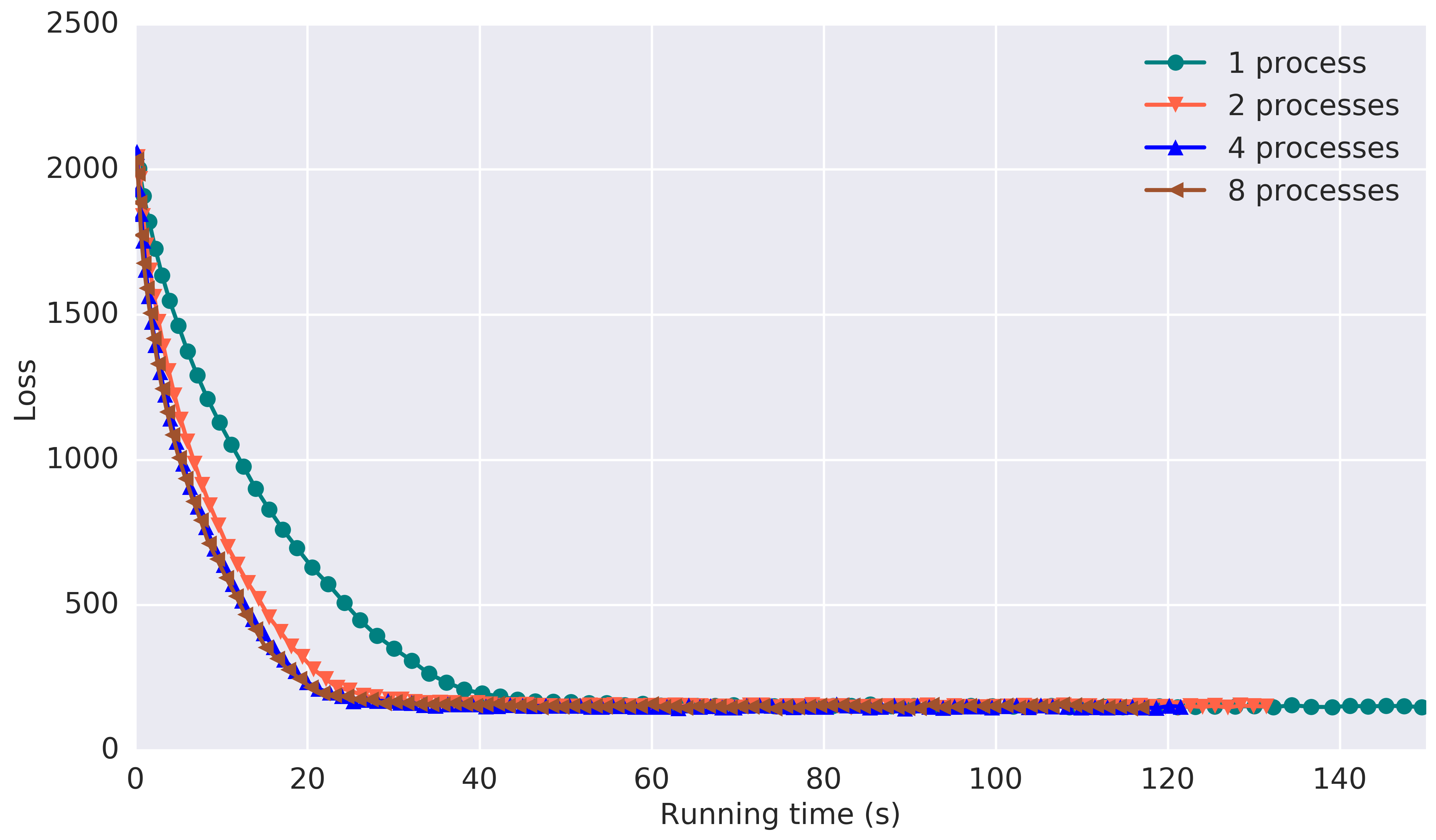}
	\caption{Convergence curves using 1, 2, 4, 8 processes}
	\label{fig:mul}
\end{figure}

All four methods converges to the same optimal solution. One can see that the 1-process method converges in $50$s, 2-process method converges in $30$s, and the 4-process method converges in $20\sim25$s. By using the multiprocessing, the solving time for the MCDODE framework can be reduced by at least a half. The marginal effects of adding processes decreases when the number of processes increases, as a result of the increasing communication costs among different processes. Since we conduct the experiments on an eight-core CPU, the eight-process method makes use of all computing resources of the CPU and hence it achieves the best convergence rate.

\subsubsection{Impact of data quantity}

In this section, we analyze the impact of number of data samples on the MCDODE framework. We solve for the baseline setting, while the number of data samples varies from $1$ to $256$. We keep track of the convergence curves for different number of data samples, the results are plotted in Figure~\ref{fig:nd}.

\begin{figure}[h]
	\centering
	\includegraphics[width=0.85\linewidth]{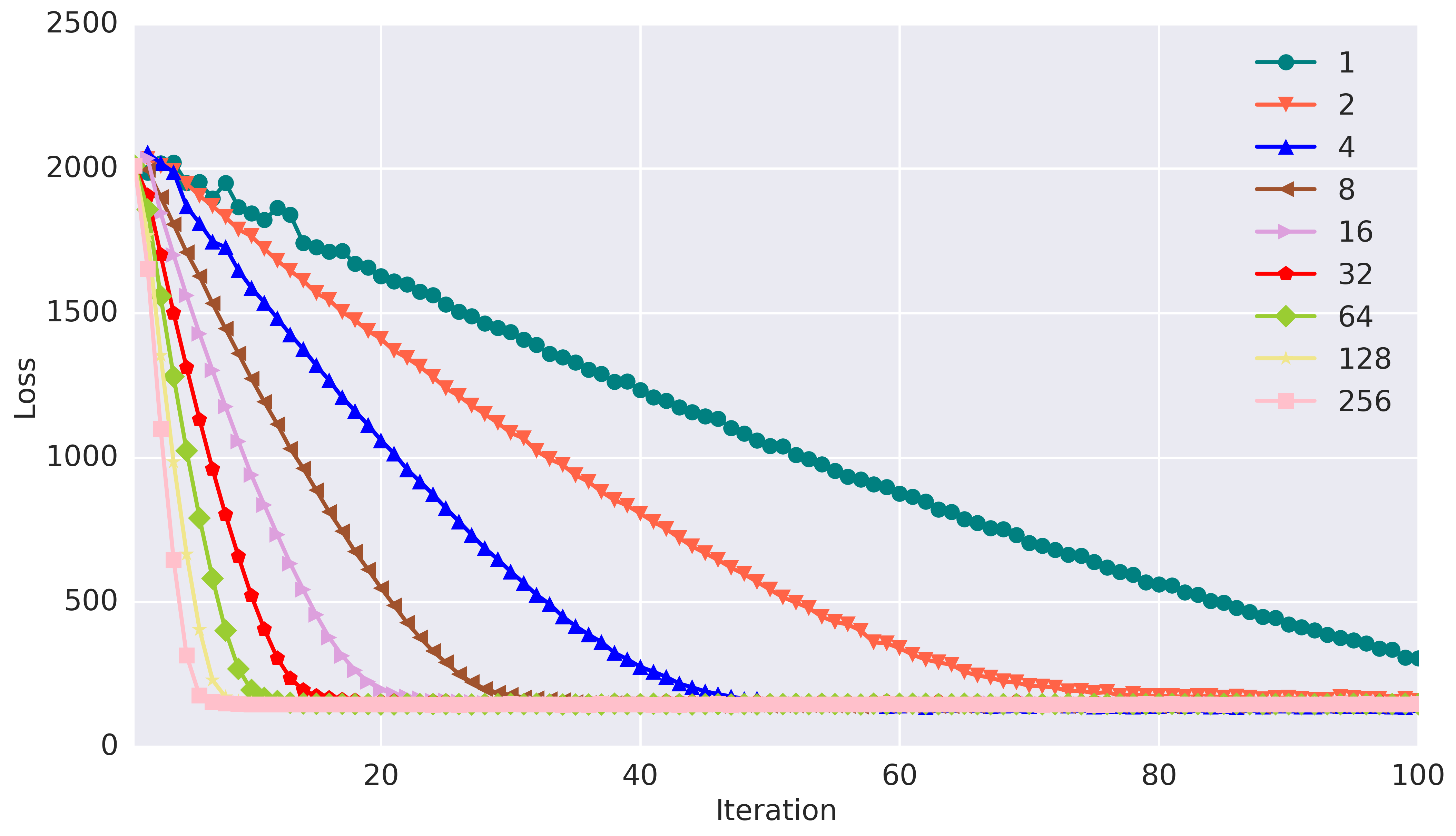}
	\caption{Convergence curves with different number of data samples}
	\label{fig:nd}
\end{figure}

As can be seen from Figure~\ref{fig:nd}, the convergence rate increases when the number of data sample increases. The solution method with one data sample does not converge in $100$ iterations while the solution method with $256$ data samples converges within $10$ iterations. In the large-scale networks, limited number of iterations can be conducted due to the lack of computational resources and time constraints, hence more data samples are usually required to ensure estimation accuracy.

\subsubsection{Impact of noise level}

We further demonstrate the impact of noise level on the proposed solution algorithm. We solve the MCDODE framework for the baseline setting, while we change the noise level from $0$ to $0.9$. The convergence curves under different noise levels are presented in Figure~\ref{fig:nl}.

\begin{figure}[h]
	\centering
	\includegraphics[width=0.85\linewidth]{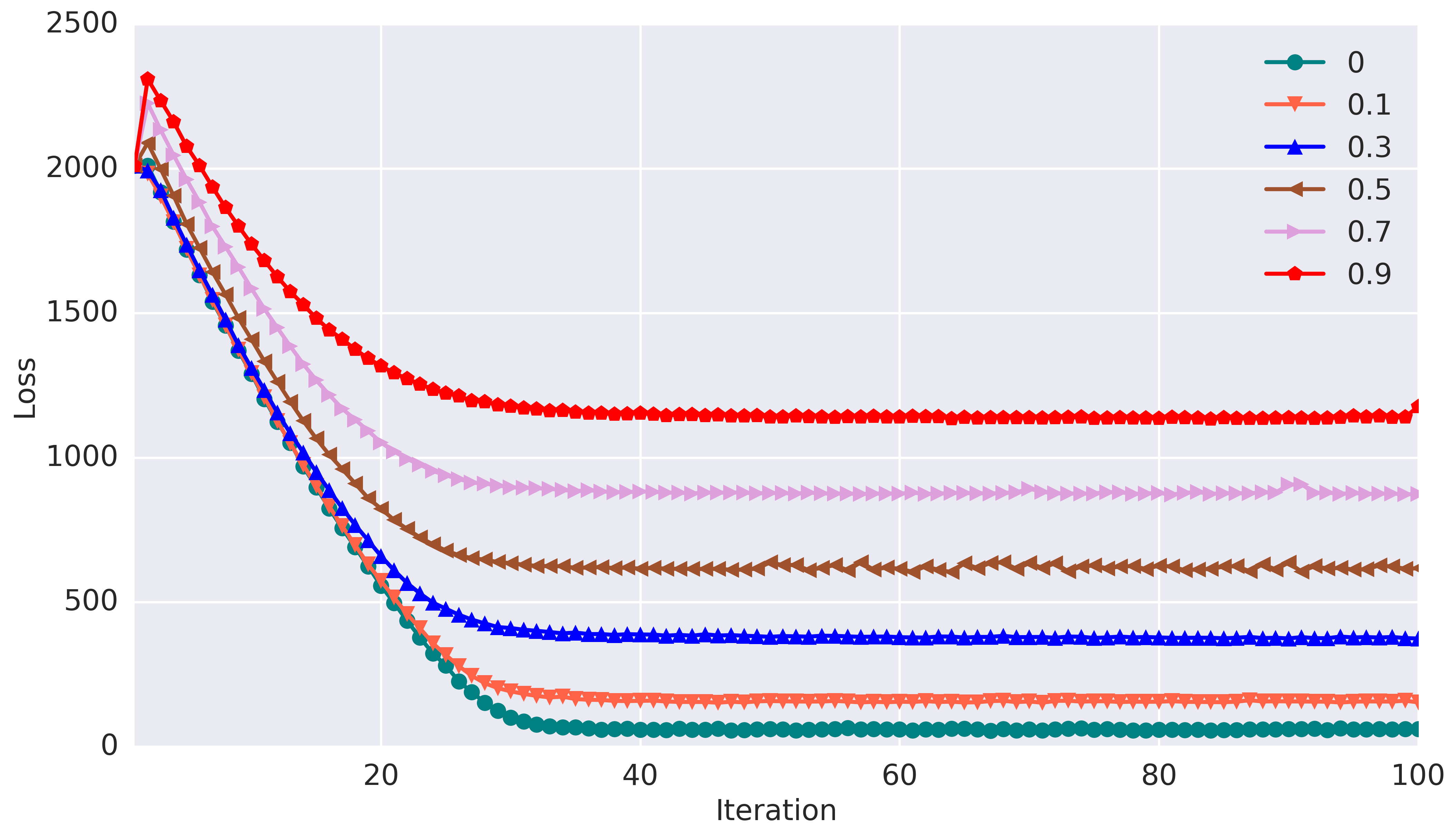}
	\caption{Convergence curves under different noise levels}
	\label{fig:nl}
\end{figure}

One can see from Figure~\ref{fig:nl} that the noise level has a significant impact on the estimation accuracy of the multi-class OD demand. When there is no noise in the observed data, the estimation method can converges to nearly zero loss in $30$ iterations. When the noise level is high, the estimation method does not converge to the optimal solution. Especially when noise level is $0.9$, the loss can only be reduced by half and the R-square for the estimated OD is around $0.7$.

\subsubsection{Sensitivity analysis}

In this section, we conduct the sensitivity analysis of the proposed MCDODE framework in terms of initial OD demand, ``true'' OD demand and step sizes.

Firstly, we perform the sensitivity analysis on the initial OD demand. We solve the MCDODE framework for the baseline setting for $100$ times. In each time, we change the initial OD demand by random generators $\texttt{Unif}(0, 15)$ for car demands and $\texttt{Unif}(0, 3)$ for truck demands. We compute the R-squares of the estimated OD demand for all $100$ estimations and the boxplot for the R-squares is presented in Figure~\ref{fig:si}.

\begin{figure}[h]
	\centering
	\includegraphics[width=0.85\linewidth]{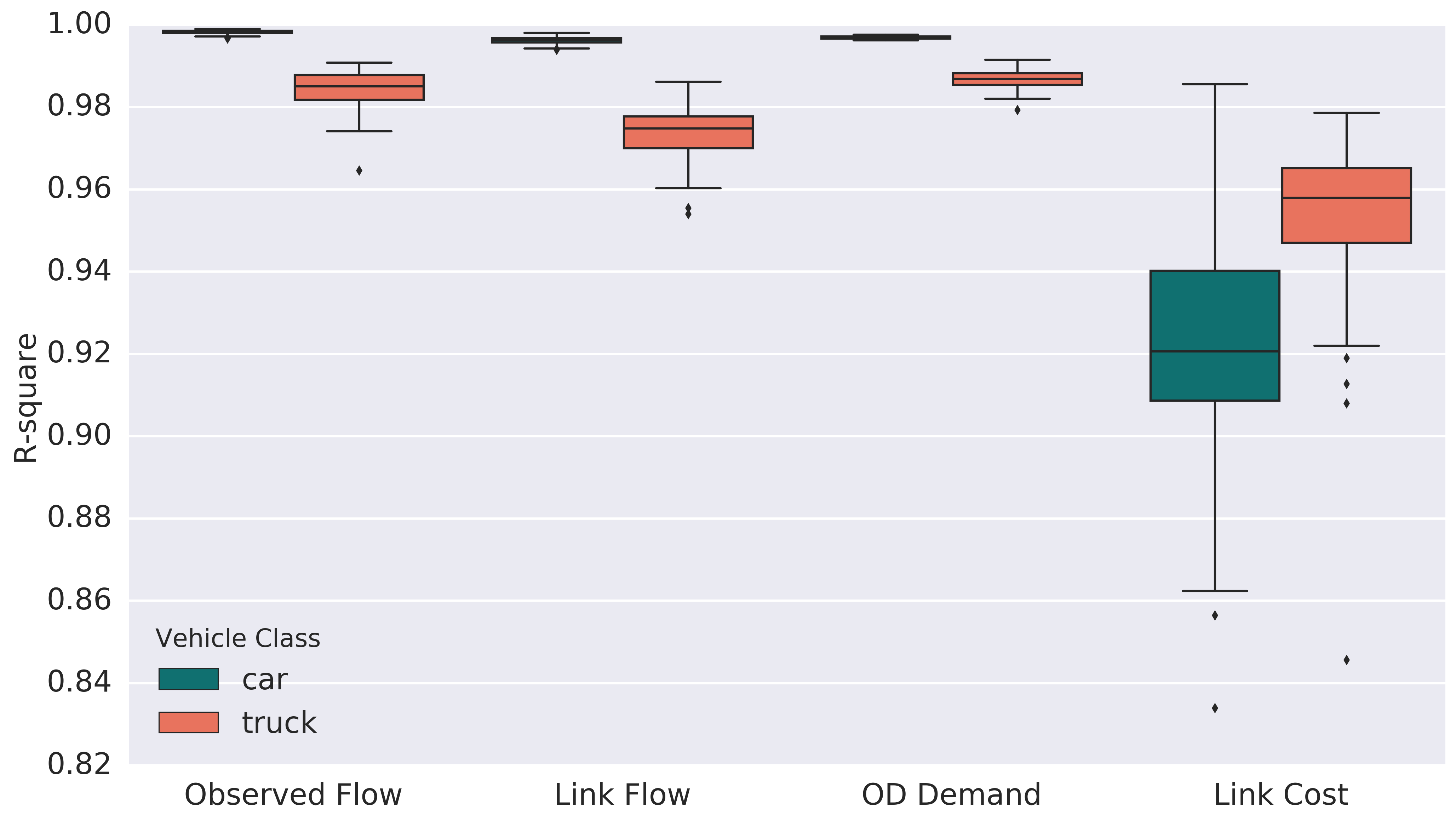}
	\caption{Boxplot for R-squares with different initial OD demand.}
	\label{fig:si}
\end{figure}

As can be seen from Figure~\ref{fig:si}, the R-squares for car flow are above $0.98$ and the variance is low for different initial OD demand, while the variance of R-squares for truck flow is high. Especially for link flow, the R-square between the ``true'' and estimated link flow for truck can be as low as $0.95$. However, all the R-squares for truck flow are still above $0.9$. The results imply that the proposed method is generally robust to the initial OD demand, but estimating the truck demand is more challenging than estimating the car demand. To ensure a satisfactory estimation result, it is desirable to run the MCDODE framework multiple times and chose the best one as the final OD demand. In contrast, the R-square of link travel time for car is lower than that for trucks, which is probably because car speeds can vary in a wide range but truck speed is relatively stable.

Secondly, we fix the initial OD demand and solve for the baseline setting for $100$ times. In each run, we sample the ``true'' OD demand from $\texttt{Unif}(0, 300)$ and $\texttt{Unif}(0,60)$ for car and truck, respectively. We compute the R-squares of estimated OD demand for the $100$ runs and the boxplot for the R-squares is presented in Figure~\ref{fig:st}.

\begin{figure}[h]
	\centering
	\includegraphics[width=0.85\linewidth]{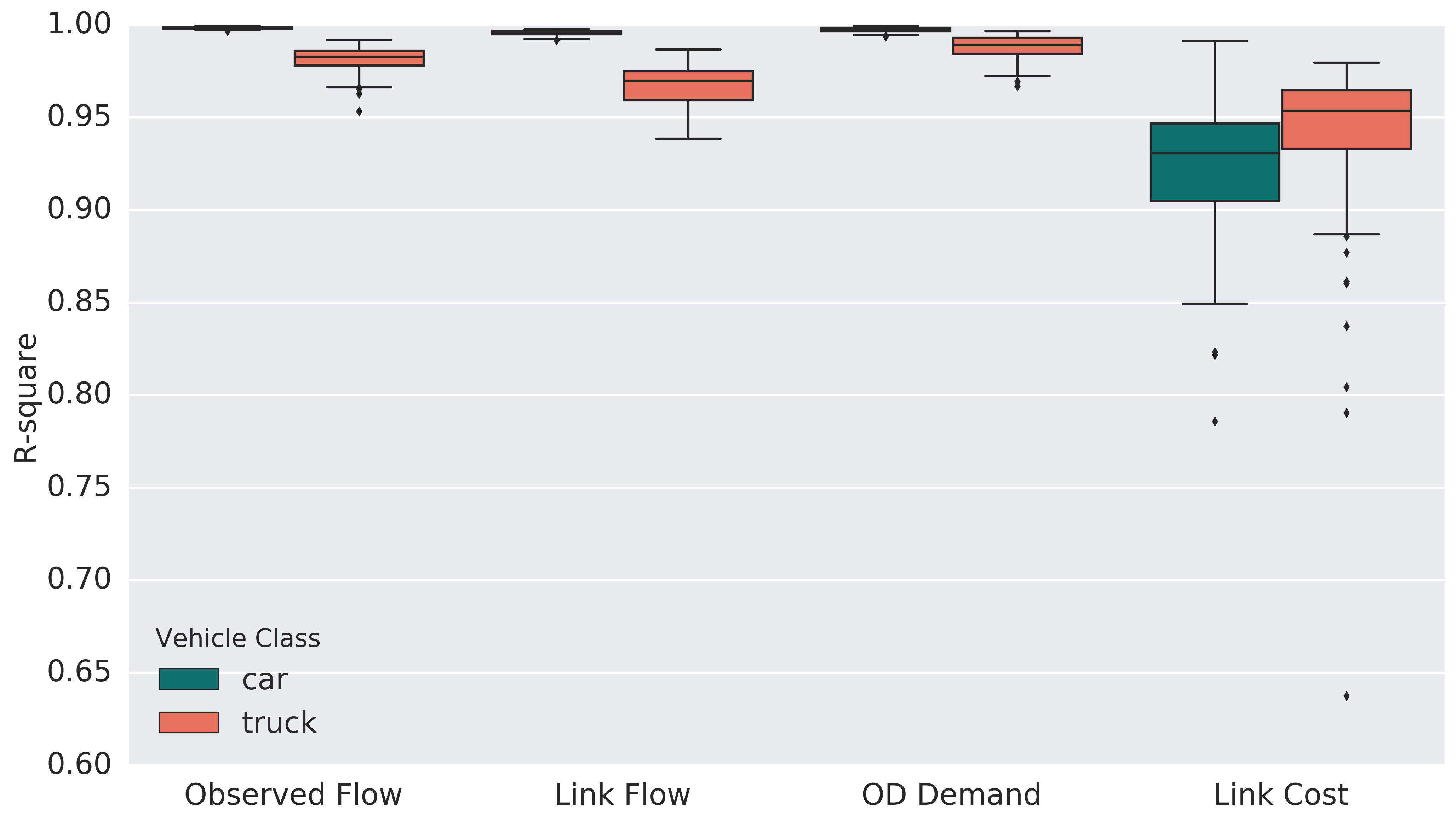}
	\caption{Boxplot for R-squares with different ``true'' OD demand.}
	\label{fig:st}
\end{figure}

The MCDODE framework achieves satisfactory accuracy for different ``true'' OD demands, and most of the R-squares for both car and truck floware over $0.9$. The R-square for truck flow is lower than that for car flow, and the variance of R-squares for truck flow is also higher. As we discussed above, estimating the truck OD demand is more challenging, as a result of the discretized and stochastic behaviors of trucks in the traffic simulation model. In addition, the R-squares of link travel time for cars are lower than that for trucks, which is similar to previous study for the initial OD demand.

Thirdly, we examine the Adagrad method with different step sizes. We solve the MCDODE framework for $7$ times under the baseline setting. In each run, we vary the step size from $0.01$ to $100$, and the convergence curves are presented in Figure~\ref{fig:ss}.
\begin{figure}[h]
	\centering
	\includegraphics[width=0.85\linewidth]{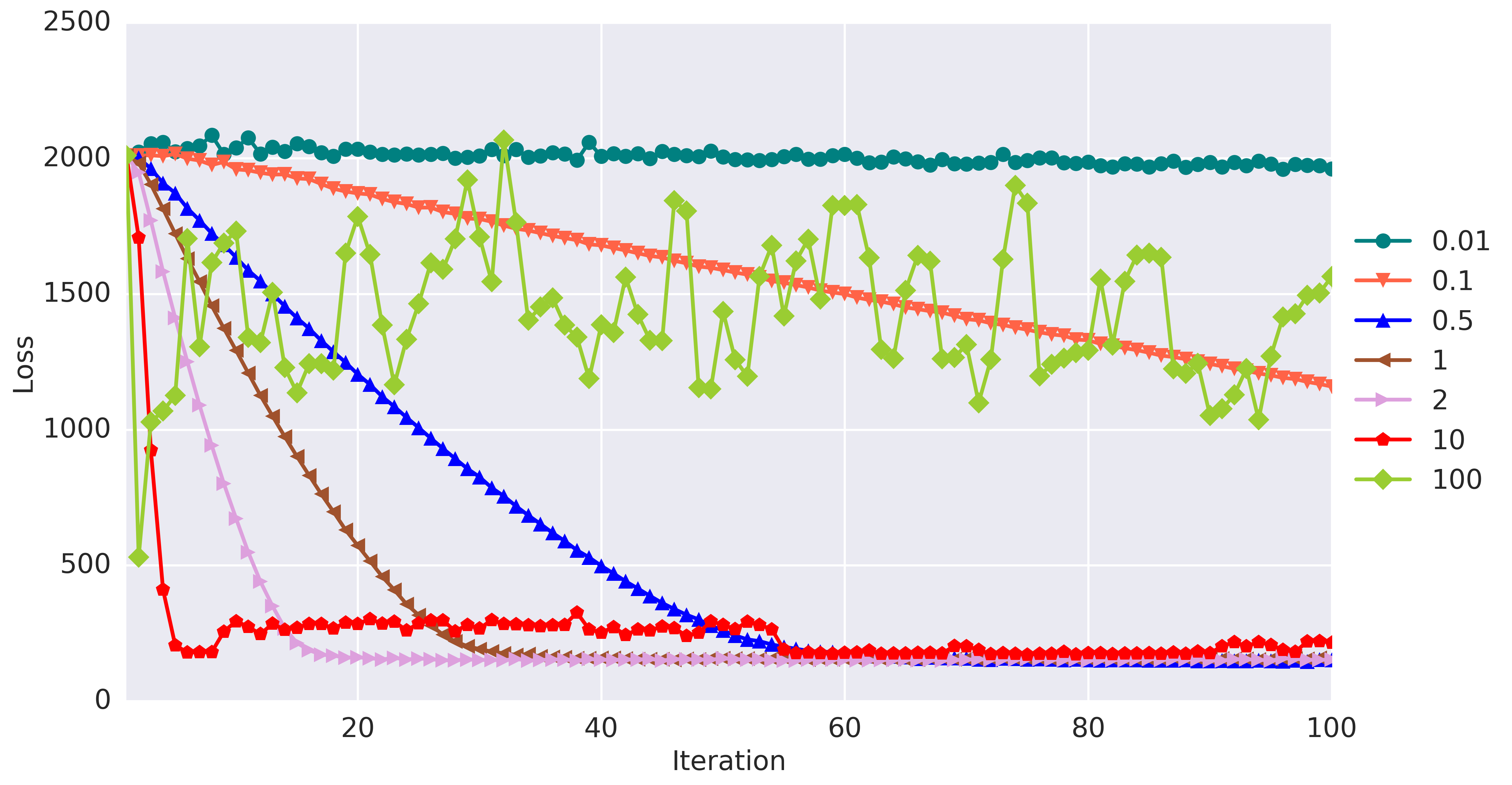}
	\caption{Convergence curves with different step sizes}
	\label{fig:ss}
\end{figure}

One can see that the Adagrad is robust to the step size, and any step size between $0.5$ and $2$ can guarantee the method to converge to the optimal solution within $60$ iterations. When the step size is too small, the method converges slowly; when the step is too large, the convergence becomes unstable and fluctuating.

\clearpage

\subsection{A large-scale network: southwestern Pennsylvania region}

In this section, we perform the MCDODE on a large-scale network for the southwestern Pennsylvania region (Figure \ref{fig:network}). The network covers ten counties of southwestern Pennsylvania region, with the Pittsburgh city located in the center. The approximate range of the 10 counties are marked by the black quadrangle in Figure \ref{fig:network}. There are around $2.57$ million population and 7,112 square miles area in the network. 
All parameters for the Pittsburgh network are listed in the Table \ref{tbl:parameters}.

\begin{figure}[h]
	\centering
	\includegraphics[width=0.90\linewidth]{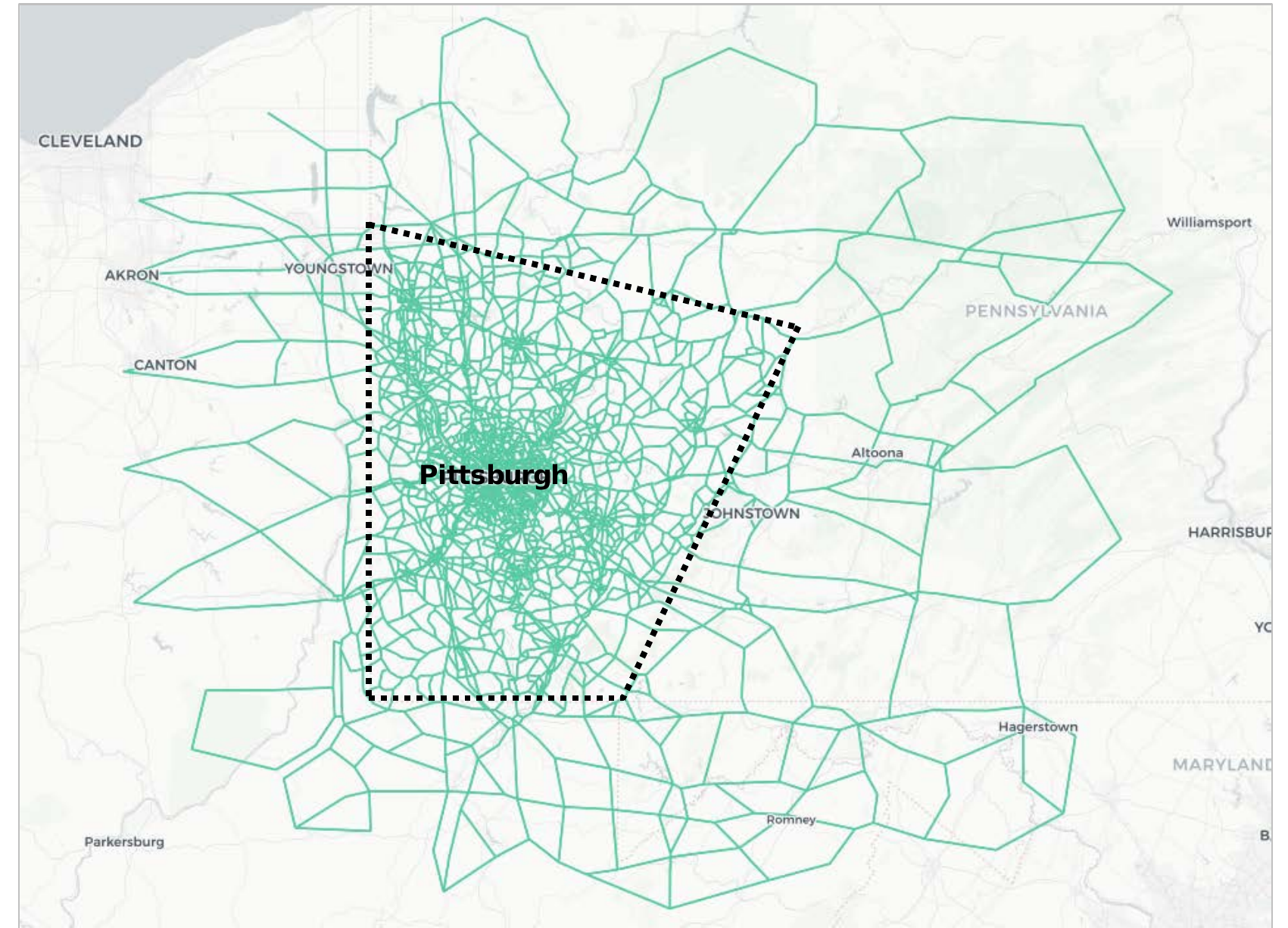}
	\caption{An overview of the network for the southwestern Pennsylvania region}
	\label{fig:network}
\end{figure}

\begin{table}[h]
	\centering
	\caption{Network parameters}
	\label{tbl:parameters}
	\begin{tabular}{@{}ll@{}}
		\hline
		Name                                       & Value              \\ \hline\hline
		studying period (weekday)                  & 6:00 AM - 11:00 AM \\
		simulation unit time & 5 seconds          \\
		Length of time interval             & 15 minutes          \\
		Number of time intervals & 30 \\
		number of links                            & 16,110             \\
		number of nodes                            & 6,297              \\
		number of origins (destinations)           & 283                \\
		number of origin-destination pairs         & 80,089             \\
		\hline
	\end{tabular}
\end{table}

We run the MCDODE framework with traffic flow data and travel time (traffic speed) data. The traffic flow data are obtained from the Pennsylvania Department of Transportation (PennDOT). The data are collected annually for some selected locations on the Pennsylvania state-owned roads, for each hour of the day and for one day of the year. The car traffic volume counts and truck traffic volume counts are collected separately, where car traffic volume counts are measured for all passenger cars and truck traffic volume counts includes all kinds of trucks at the measured location. The traffic count can be either one-directional or bi-directional. Since all count data are measured in hours, in data pre-processing we smooth the hourly count data to 15-minute interval. There are $608$ locations in total that has valid car and truck volume counts for MCDODE. Traffic speed data are obtained Federal Highway Administration (FHWA) for the year 2016. The speed data are observed every 5-minutes of the day for highway segments, the data are also classified to cars and trucks. We aggregate the travel time data to 15-minute interval. In total,
there are 945 locations with valid car and truck travel time observations.

The initial OD demands for cars and trucks are randomly generated from a uniform distribution $\texttt{Unif}(0, 0.01)$ and $\texttt{Unif}(0,0.001)$, respectively. We aggregate all the traffic flow and travel time observations to a single data sample and use the single-process Adagrad method to solve MCDODE. The step size is $e^{-5}$. We set $w_1= 1, w_2=0.01$. We use the hybrid dynamic traffic assignment model as the route choice model \citep{qian2013hybrid}, the adaptive ratio is $0.2$ for cars and $0$ for trucks. The MCDODE framework runs for $55$ iterations. In first $35$ iterations, both OD demand for car and trucks are updated simultaneously; while for the last $20$ iterations, we only update the truck demand since its more challenging to estimate.

The convergence of the proposed MCDODE framework is presented in Figure~\ref{fig:lc}. Overall, the solution algorithm performs well and the objective function converges fast. It takes around $25 \times 55 = 1375$ minutes (around $23$ hours) to complete the $55$ iterations. For each iteration, the traffic simulation takes around $7$ minutes, and the other $18$ minutes are used for the demand estimation. Constructing DAR matrix is identified as the bottleneck in the demand estimation.

\begin{figure}[h]
	\centering
	\includegraphics[width=0.7\linewidth]{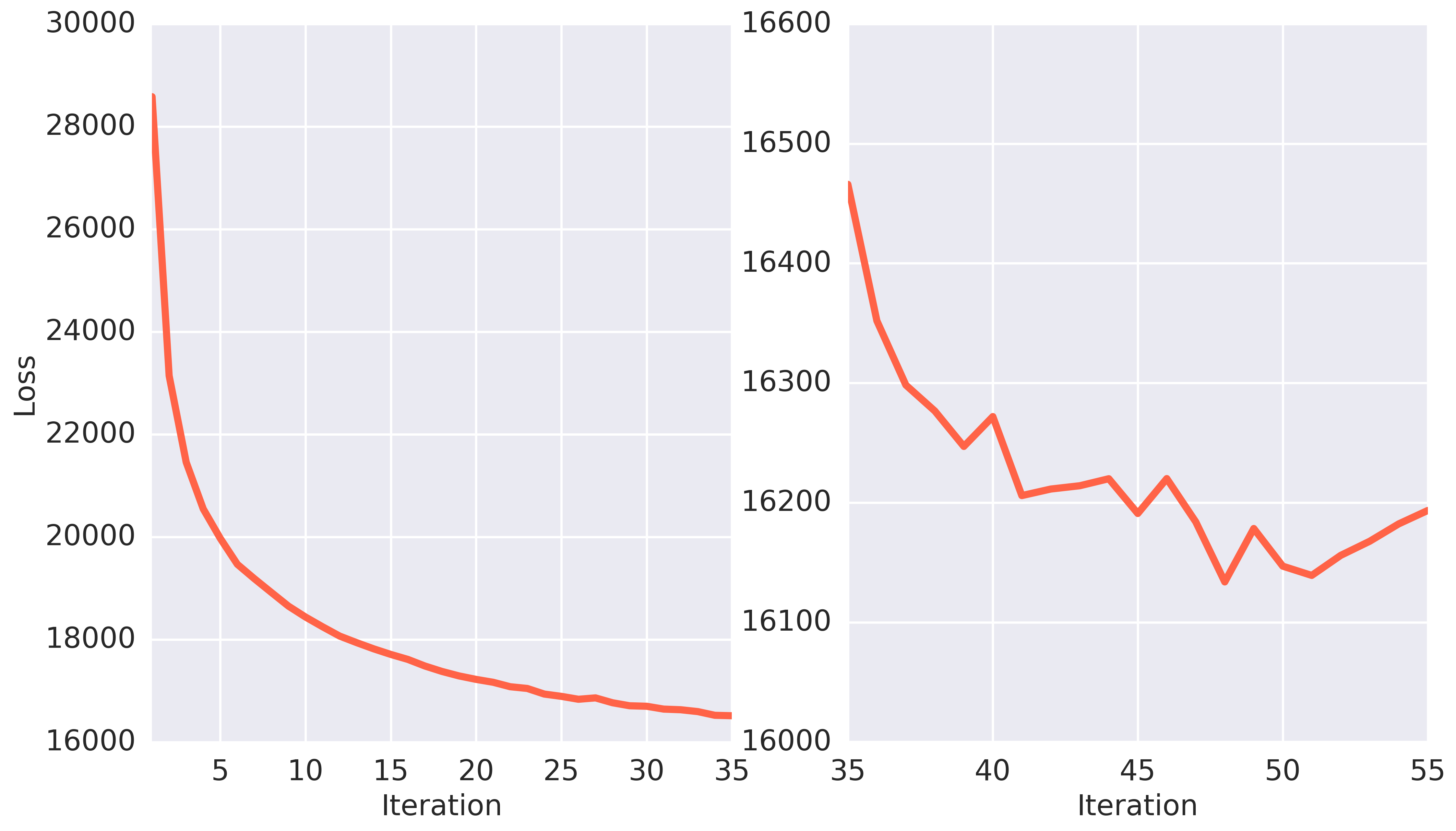}
	\caption{Convergence of the objective function for 55 iterations \footnotesize{(OD demands for cars and trucks are updated simultaneously in first $35$ iterations, while only truck OD demand is updated in last 20 iterations.)}}
	\label{fig:lc}
\end{figure}

The comparisons between the observed and estimated flow are presented in Figure~\ref{fig:lgf}, and R-squares are $0.66$ and $0.59$ for car and truck, respectively. One can see that the truck flow is roughly one tenth of the car flow, and the estimation accuracy for car flow is higher than truck flow.  Overall, the results of the MCDODE framework is compelling and satisfactory for such a large network.

\begin{figure}[h]
	\centering
	\includegraphics[width=0.7\linewidth]{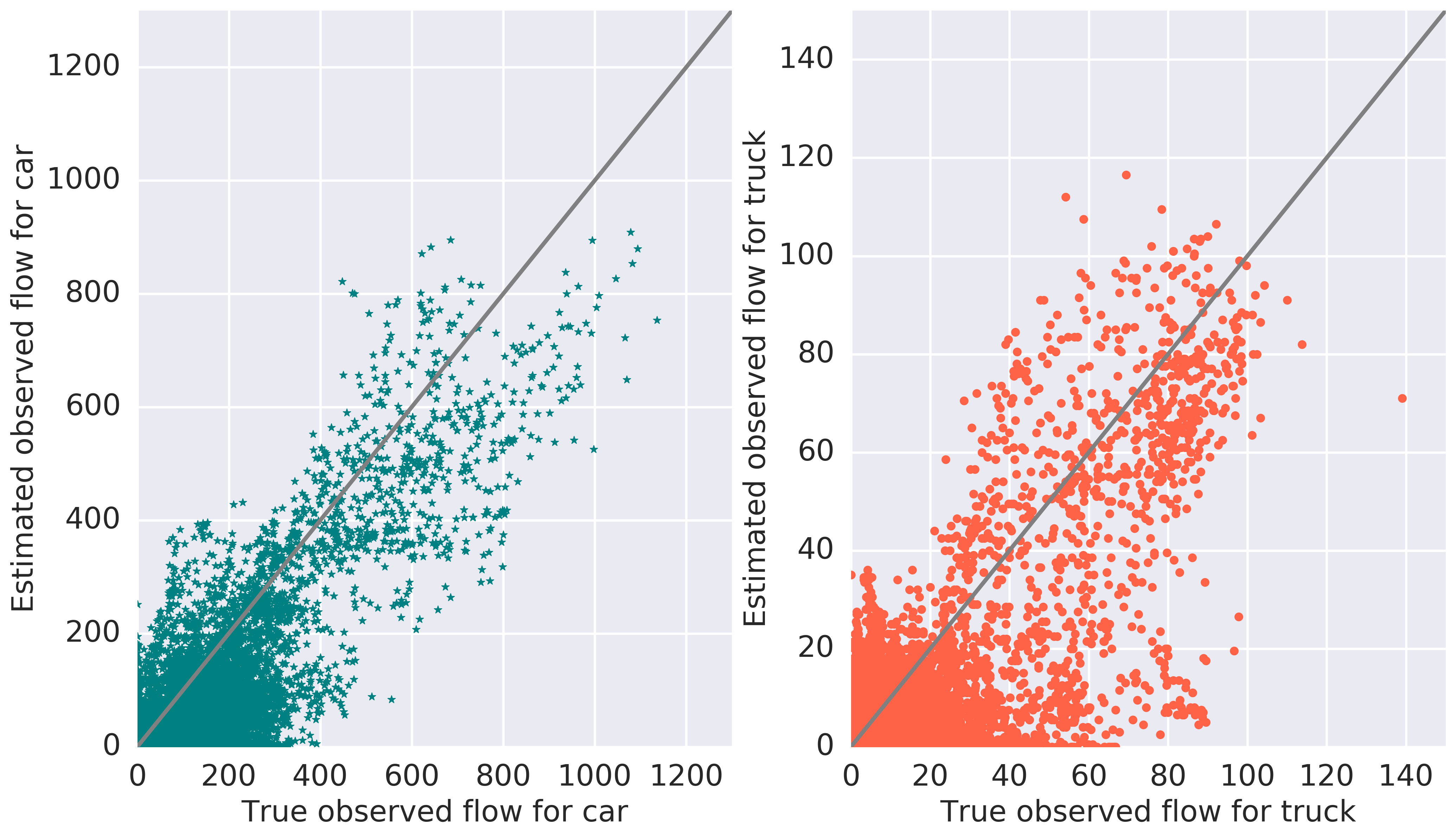}
	\caption{Estimated and observed flow for cars and trucks (unit:vehicle/15mins)}
	\label{fig:lgf}
\end{figure}

We also plot the comparisons between the observed and estimated link speed in Figure~\ref{fig:lgs}. The reason we plot the link speed instead of link travel time is that the link travel time can vary in a wide range for large scale networks, but link speed usually varies between $20$ to $80$ miles/hour. Hence, visualizing the link speed is more straightforward and legible.

\begin{figure}[h]
	\centering
	\includegraphics[width=0.7\linewidth]{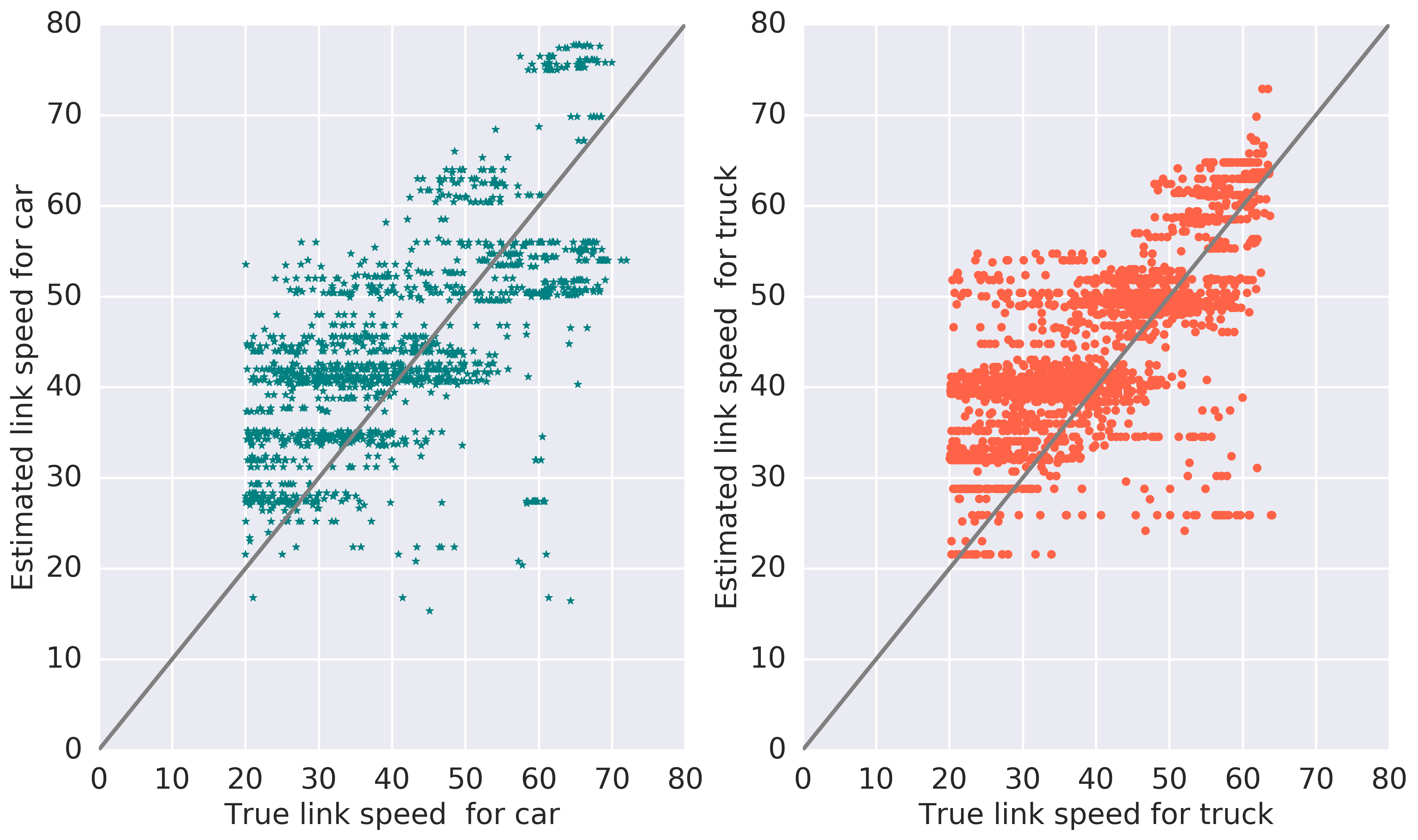}
	\caption{Estimated and observed flow for cars and trucks (unit:miles/hour)}
	\label{fig:lgs}
\end{figure}

The  R-squares between the observed and estimated link speed are $0.45$ and $0.51$ for car and truck, respectively. One can see that the estimation accuracy of traffic speed is not as good as the traffic flow, which is probably attributed to unsatisfactory approximations of the link travel time derivatives $\frac{\partial\bar{\Lambda}(\{\vec{x}_i\}_i)}{\partial \vec{x}_i}$ in dynamic networks. Improving the approximation quality for the link travel time will be left for the future research.

In addition, the estimation uses the randomly generated OD demand as the initial point for the MCDODE framework. If prior knowledge of the OD demand can be obtained from traditional planning models, the estimation accuracy might further be improved. The estimated multi-class OD demand is further used to study the impact of a development project in Pittsburgh, and details are presented in \citet{CARTRUCK}.
\clearpage

\section{Conclusion}
\label{sec:con}
This paper presents a data-driven framework for multi-class dynamic origin-destination demand estimation (MCDODE) using observed traffic flow and travel time data. The traffic data can be any linear combinations of flow characteristics (e.g. counts, time or travel time) across vehicle classes, road segments and time intervals. All the characteristics involved in the MCDODE formulation are vectorized, and the proposed framework is represented on a computational graph. The computational graph can be solved efficiently through a forward-backward algorithm for large-scale MCDODE problems. In the forward iteration, the dynamic traffic assignment problem is solved, and the loss (objective function) is computed through a series of equations.
In the backward iteration, the OD demand is updated by the backpropagation method with the route choice matrix, DAR matrix  and route travel time known from the forward iteration. The MCDODE formulation is solved when the forward-backward algorithm converges.

Practical issues related to MCDODE framework are discussed. We adopt a mesoscopic multi-class traffic simulation package MAC-POSTS to solve for the spatio-temporal path/link flow. The DAR matrix is highly sparse, and thus we propose novel tree-based cumulative curves from MAC-POSTS to construct the sparse DAR matrix. We incorporate multi-day observation data to the MCDODE framework, and different variants of gradient-based solution algorithms are discussed and compared.

The proposed MCDODE framework is examined on a small network as well as a real-world large-scale network. The objective function converges quickly with the Adagrad method. We also conduct the sensitivity analysis of the estimation accuracy with respect to initial OD demand, ``true'' OD demand, and step sizes. Overall, the estimation results are compelling, satisfactory and robust, and the forward-backward algorithm is computationally plausible for large-scale networks. The estimated multi-class dynamic OD demand can help policymakers to better understand the dynamics of OD demand and the spatio-temporal distribution of vehicles in terms of different classes.

In the near future, we plan to improve the estimation accuracy of the MCDODE framework in the following two directions: 1) some prior knowledge about the dynamic OD demand can be used as the initial point for the solution methods. For example, we can construct the DAR matrix from the speed data directly and estimate the dynamic OD without running the simulation. The estimated OD demand can be used as the initial point for the proposed MCDODE framework \citep{ma2018estimating}; 2) the method of approximating the derivatives of link travel time under multi-class flow can be further improved. In addition, we plan to extend this research to calibrate the parameters in route choice models as well as the fundamental diagrams, thanks to the versatile computational graph framework. We also plan to extend this research to estimate the probabilistic distribution of multi-class dynamic OD demand and explore the spatio-temporal characteristics of dynamic OD demand.

\section*{Supplementary materials}
The mesoscopic multi-class traffic simulation package MAC-POSTS \footnote{\url{https://github.com/Lemma1/MAC-POSTS}} and the MCDODE framework \footnote{\url{https://github.com/Lemma1/MAC-POSTS/tree/master/src/examples/mcDODE}} are implemented and open-sourced on Github.

\section*{Acknowledgments}
This research is supported by NSF Award CMMI-1751448, and Carnegie Mellon University's Mobility 21, a National University Transportation Center for Mobility sponsored by the US Department of Transportation.

\cleardoublepage
\bibliography{report}
\cleardoublepage

\end{document}